\documentclass[acmsmall,nonacm,screen,authorversion]{acmart}

\AtBeginDocument{%
  \providecommand\BibTeX{{%
    \normalfont B\kern-0.5em{\scshape i\kern-0.25em b}\kern-0.8em\TeX}}}

\usepackage{caption}
\usepackage{subcaption}
\usepackage{hyperref}
\usepackage{bm}
\usepackage{booktabs}
\usepackage{multirow}
\usepackage[irlabel]{fcps}

\usepackage{tikz}
\usetikzlibrary{positioning}
\usetikzlibrary{shapes.multipart}
\usetikzlibrary{arrows}

\usepackage{prettyref}
\newcommand{\rref}[2][]{\prettyref{#2}}
\newrefformat{sec}{Section\,\ref{#1}}
\newrefformat{app}{Appendix\,\ref{#1}}
\newrefformat{def}{Def.\,\ref{#1}}
\newrefformat{thm}{Theorem\,\ref{#1}}
\newrefformat{prop}{Proposition\,\ref{#1}}
\newrefformat{lem}{Lemma\,\ref{#1}}
\newrefformat{cor}{Corollary\,\ref{#1}}
\newrefformat{rem}{Remark\,\ref{#1}}
\newrefformat{ex}{Example\,\ref{#1}}
\newrefformat{tab}{Table\,\ref{#1}}
\newrefformat{fig}{Fig.\,\ref{#1}}
\newrefformat{case}{case\,\ref{#1}}
\newrefformat{foot}{Footnote\,\ref{#1}}
\newrefformat{line}{line\,\ref{#1}}
\newrefformat{model}{Model\,\ref{#1}}

\newcommand{\rp}{\ensuremath{r_p}}
\newcommand{\hp}{\ensuremath{h_p}}

\newcommand{\imagescale}{0.23}

\floatstyle{ruled}
\newfloat{model}{tbp}{model}
\floatname{model}{Model}
\newcounter{modelline}
\makeatletter
\newcommand{\mline}[1]{\refstepcounter{modelline}\ltx@label{#1}\quad\text{\scriptsize{\themodelline}}\quad}
\makeatother

\usepackage{ifpdf}
\ifpdf
\fi

\begin{document}

\title{Formally Verified Next-Generation Airborne Collision Avoidance Games in ACAS~X}
\author{Rachel Cleaveland}
\email{rcleavel@andrew.cmu.edu}
\orcid{0000-0002-6306-9502}
\author{Stefan Mitsch}
\email{smitsch@cs.cmu.edu}
\orcid{0000-0002-3194-9759}
\author{Andr{\'e} Platzer}
\email{aplatzer@cs.cmu.edu}
\orcid{0000-0001-7238-5710}
\affiliation{%
  \institution{Carnegie Mellon University}
  \streetaddress{5000 Forbes Avenue}
  \city{Pittsburgh}
  \state{PA}
  \postcode{15213}
  \country{USA}}

\begin{abstract}
The design of aircraft collision avoidance algorithms is a subtle but important challenge that merits the need for provable safety guarantees. Obtaining such guarantees is nontrivial given the unpredictability of the interplay of the intruder aircraft decisions, the ownship pilot reactions, and the subtlety of the continuous motion dynamics of aircraft. Existing collision avoidance systems, such as TCAS and the Next-Generation Airborne Collision Avoidance System ACAS~X, have been analyzed assuming severe restrictions on the intruder's flight maneuvers, limiting their safety guarantees in real-world scenarios where the intruder may change its course.

This work takes a conceptually significant and practically relevant departure from existing ACAS~X models by generalizing them to hybrid games with first-class representations of the ownship and intruder decisions coming from two independent players, enabling significantly advanced predictive power. By proving the existence of winning strategies for the resulting Adversarial ACAS~X in differential game logic, collision-freedom is established for the rich encounters of ownship and intruder aircraft with independent decisions along differential equations for flight paths with evolving vertical/horizontal velocities. We present three classes of models of increasing complexity: single-advisory infinite-time models, bounded time models, and infinite time, multi-advisory models. Within each class of models, we identify symbolic conditions and prove that there then always is a possible ownship maneuver that will prevent a collision between the two aircraft.
\end{abstract}



\maketitle

\providecommand{\bebecomes}{\mathrel{::=}}
\providecommand{\alternative}{~|~}
\providecommand{\ctrl}{\ensuremath\textit{ctrl}}
\providecommand{\plant}{\ensuremath\textit{plant}}

\section{Introduction}\label{sec:introduction}

Mid-air aircraft collisions are a fundamental responsibility of pilots and air traffic controllers to avoid, but their likelihood only increases as air space gets more congested and Unmanned Aerial Vehicles become more prevalent. The first onboard collision avoidance system, known as Traffic Alert and Collision Avoidance System (TCAS), was developed in the 1970s and has successfully prevented several mid-air collisions. However, this system is not perfect; one particular failure of TCAS occurred in the 2002 {\"U}berlingen crash, where two airplanes collided despite having received instructions by their TCAS systems onboard. Tragedies like this underscore the importance of continued research into developing and formally verifying onboard collision avoidance systems. 

Most of the time, when aircraft are on a collision course, they are detected and resolved in advance by either the pilots or flight directors of the Air Route Traffic Control Centers. However, in rare scenarios where conflicting flight paths were not detected early enough and two aircraft are on an immediate collision course, collision avoidance maneuvers must be performed as a last resort. With little time to determine and perform the necessary maneuvers to avoid collision, it is imperative to ensure the safety of these collision avoidance maneuvers in advance using formal verification under all reasonable flight circumstances to ensure that no midair collisions happen. 

The TCAS, and the more recent ACAS~X, collision avoidance systems developed by the Federal Aviation Association (FAA) give vertical ascent/descent advisories when an aircraft is encountering an intruder with which it is at risk of colliding \cite{FAA2}. The goal of ACAS~X is to prevent Near Mid-Air Collisions (\emph{NMACs}), dangerous situations where two aircraft come within $r_p = 500$\,ft horizontally and $h_p = 100$\,ft vertically of each other \cite{kochenderfer}. These variables $r_p$ and $h_p$ describe the radius and height, respectively, of a puck surrounding the aircraft, into which no other aircraft should enter.

Previous work explores formal verification of ACAS~X when the intruder aircraft is moving at a constant horizontal and vertical velocity \cite{acasx}. This assumption is rigid and does not take into account the potential maneuvers that the intruder may perform.
This article takes a conceptually significant departure by generalizing formally verified ACAS~X models from hybrid \emph{systems} to \emph{hybrid games}, owing to the fundamental observation that, despite best intent, the actions of the ownship and intruder aircraft may interfere with one another since they are resolved by different pilots with different situational awareness facing a challenging safety hazard.
This generalization is of practical relevance for the predictive power of verified ACAS~X models but requires a fundamental shift in reasoning using differential game logic for hybrid games \cite{dGL2,dGL3,Platzer18}.
Hybrid systems are fundamentally single player. Only hybrid games can faithfully represent a dynamics where different pilots of different aircraft may independently reach different decisions at different times with different consequences on the flight of the two aircraft.
While all pilots share the intent of avoiding collisions, only hybrid games accurately reflect that their decisions may, nevertheless, interfere, because the pilots chose different means to avoid collisions that may conflict.

\subsection{Airborne Collision Avoidance System ACAS~X} \label{sec:ACASXSec}

ACAS~X tracks the position and velocity of the ownship and intruders in its vicinity using a variety of sensors to compute its collision avoidance advisories \cite{FAA1}. An advisory alerts the pilot with an audio-visual message and requests that she either maintain her vertical speed, or accelerate towards a new desired vertical speed. An advisory is issued only when a potential collision is identified, otherwise the system stays quiet to avoid distracting the pilot \cite{FAA2}. Advisories apply only in the vertical direction, not the horizontal direction, and only apply to the aircraft's climb rate. 

\begin{table}[tbhp]
    \centering
    \caption{ACAS~X advisories and their parameters as summarized in \cite{acasx}}%
    \label{tab:1}%
    \includegraphics[scale=0.47]{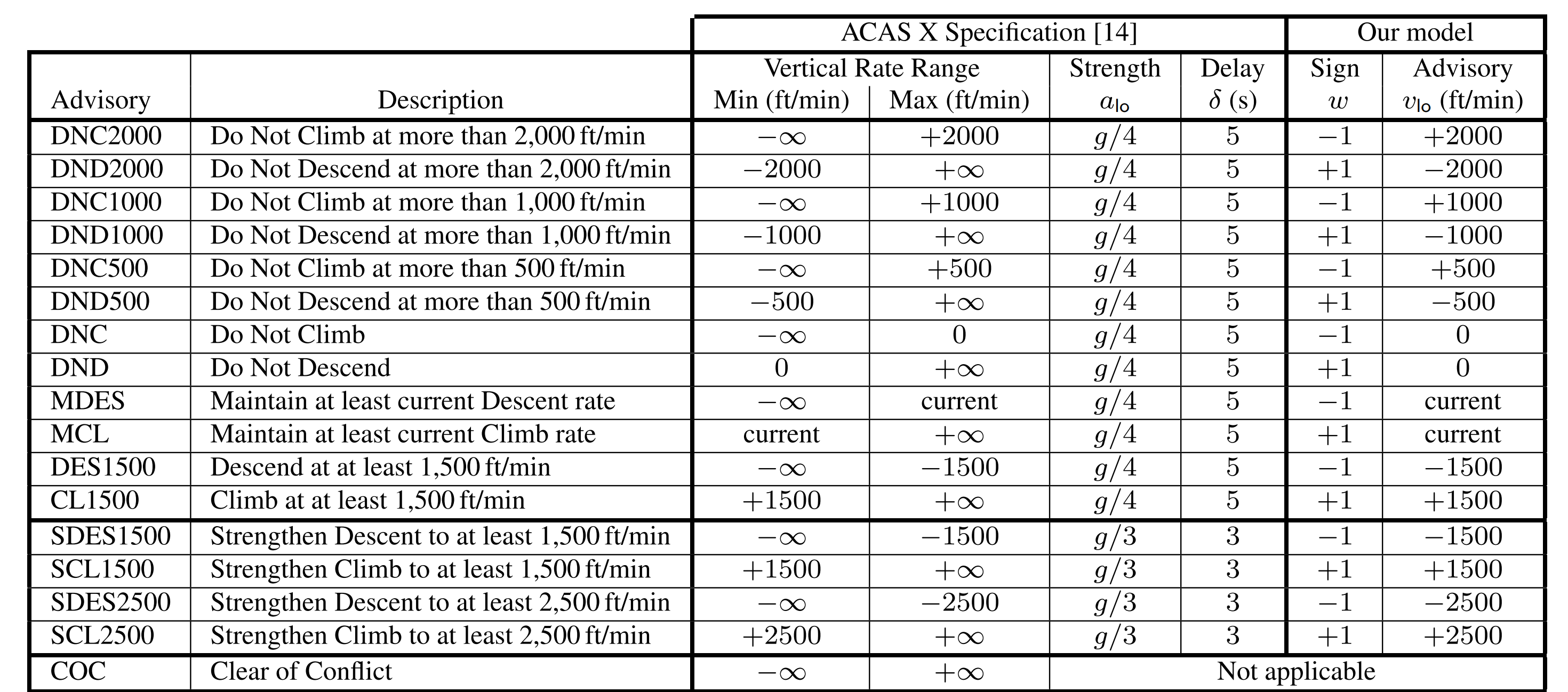}%
\end{table}

Table~\ref{tab:1} gives all of the 16 possible advisories issued by ACAS~X, plus Clear-of-Conflict (COC), which indicates that no action is necessary. These advisories vary in the extremeness of the action; a less extreme advisory like Do Not Descend (DND) only requires that the ownship does not, as the name suggests, descend past its current altitude. A more extreme advisory like SCL2500 requires the ownship to reach a climb rate of at least 2,500\,ft/min. Advisories can also be either lower bounds, like SCL2500, or upper bounds, like DNC2000 which requires that the ownship not exceed a climb rate of more than 2000\,ft/min. 
The FAA assumes that in order for the pilot to achieve the desired climb or descend rate, she does so by following a vertical acceleration of strength at least $g/4$ (referred to as the positive constant $a_{\text{lo}}$) \cite{MDP}, an assumption which will be pertinent later. 

These advisories result from an estimation of the pilot's optimal course of action, calculated by linearly interpolating a precomputed table of scores for various actions. The domain of this table includes parameters describing the state of the encounter, while its range gives scores for each possible action \cite{MDP}. This table is constructed from a Markov Decision Process which approximates the dynamics of the system on a discretized grid of the state space. From there, dynamic programming is used to optimize the table through maximizing the expected value of each event over all future outcomes for each action \cite{MDP}. These expected values approximately map to different outcomes: Near Mid-Air Collisions (NMACs), for example, correspond to large negative values, while issuing advisories corresponds to small negative values. The ACAS~X system then uses a multilinear interpolation of grid points and heuristics to choose the action with the greatest expected value given the particular circumstances surrounding the ownship's current flight conditions.

\subsection{Formally Verified Safe Regions and Hybrid Game Logic} \label{sec:hybridgamelogic}

Previous work~\cite{acasx2,acasx} applied hybrid systems to the formal verification process, a natural application given the combination of discrete advisories and continuous dynamics of an aircraft using ACAS~X. While direct verification of the ACAS~X implementation is infeasible given the complexity of ACAS~X (whose core lookup table defines 29,212,664 interpolation regions in a 5-dimensional state-space giving rise to at least half a trillion cases to consider), this work cut down the complexity with the concept of \emph{safe regions}. A region is proven safe if for all possible ownship positions and velocities within the region, an NMAC with the intruder will never occur. Thus, if ACAS~X issues an advisory, and following this advisory in any permitted way always keeps the ownship within our safe region, then this advisory is guaranteed to maintain safety. These regions comprise fully symbolic parameters like $a_{\text{lo}}$, making them easily adaptable to new ACAS~X versions.

In this work, we continue with the identification of safe regions to prove safety of the overall system, but we make the important change of applying \emph{hybrid games}, rather than \emph{hybrid systems}, in our formal verification of ACAS~X. This change is motivated by the goal to model scenarios in which the intruder is maneuvering, such as being able to change its horizontal direction or vertical velocity. Where hybrid systems only allow one actor in the system to resolve decisions, hybrid games give multiple actors \emph{independent} decision-making ability.

More specifically, the models presented in this article are rephrased using Differential Game Logic $\dGL$~\cite{dGL2,dGL3,Platzer18} from the Differential Dynamic Logic $\dL$~\cite{dL1,dL2,dL3,dL4,Platzer18} of previous work. $\dGL$ is an extension of $\dL$, so it also supports discrete assignments, control structures, and following of differential equations to represent pilot decisions, trajectory requirements, and aircraft dynamics, respectively. However, $\dGL$ can also represent adversarial dynamics, meaning $\dGL$ can express two different players in a game scenario making independent decisions that may interfere.
We make crucial use of this multi-player dynamics in our ACAS~X game models in order to enable both aircraft to maneuver independently.
Contrast this flexibility with hybrid systems models of ACAS~X~\cite{acasx2,acasx}, which are necessarily limited to a single fixed policy for the intruder (the intruder cannot  maneuver but is assumed to follow a straight line trajectory in prior ACAS~X work~\cite{acasx2,acasx}).
 
In the context of collision avoidance, one can think of the ownship as being a good-faith actor attempting to avoid collision, while the intruder is able to act independently in ways that, perhaps out of confusion, may interfere with the safety of the system. Since these two players follow independent intent, $\dGL$ works perfectly in this scenario to express these adversarial dynamics.

\begin{figure}
    \centering
    \includegraphics[width=\columnwidth]{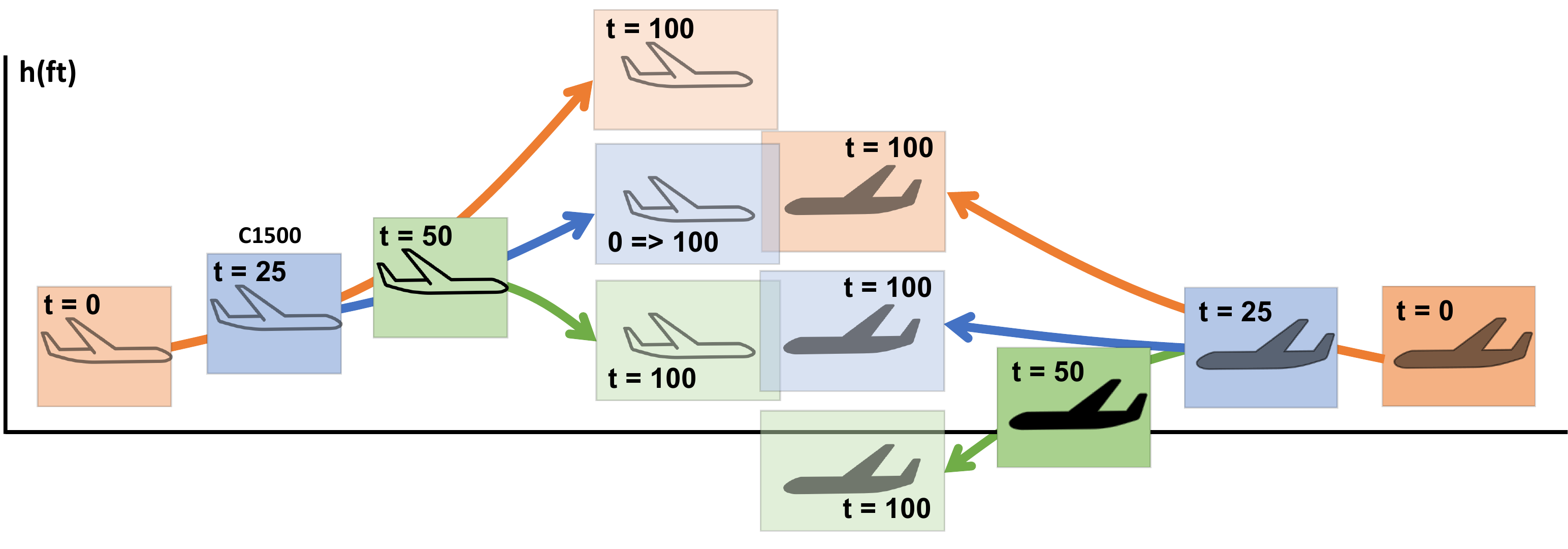}
    \Description[Illustrates ownship reactions to intruder behavior of varying cooperativeness]{Illustrates the reactions of the ownship over encounters with intruders of varying cooperativeness: reactions to an intruder that interferes the ownships path vs. is somewhat cooperative vs. helps resolve the conflict.}
    \caption{An encounter between the ownship and example reactions to an intruder that interferes (orange trajectories), is somewhat cooperative (blue trajectories), or helps resolve the conflict (green trajectories).}
    \label{fig:encounter-overview}
\end{figure}

Of course, in the real world intruder will not actively attempt to collide with the ownship, but if the ownship's goal is to avoid collision no matter the actions of the intruder, it is important to consider even the worst-case maneuvers that the intruder may perform.
Notably, our ACAS~X game model considers the case where the intruder's actions \emph{may} interfere with the safety of the system but does not assume they will.
Indeed, the actions that the ownship pilot's winning strategies for the ACAS~X game needs to take to avoid collision are less extreme when the intruder pilot reaches helpful decisions and more extreme otherwise, see \rref{fig:encounter-overview}.
$\dGL$ is implemented in the theorem prover \KeYmaeraX~\cite{kyx}, with which we verify our safe regions with respect to our models. 

As far as we know, this is the first work to apply hybrid games to the problem of aircraft collision-avoidance. Hybrid games enrich the fidelity of the safety analysis for collision-avoidance algorithms, because they capture the important phenomenon that the respective pilots of intruder and ownship aircraft reach their decisions independently, while, at the same time, being faithful to the advisories of ACAS~X. 
Given the range of trajectories that either pilot could follow as they react to their mutual responses during an encounter, a game theory perspective on collision-avoidance greatly expands the scenarios which can be modeled and proven. 

The article is organized as follows. In \rref{sec:modelingapproachSec}, we give an overview of the structure that the models have in common. In Sections~\ref{sec:InfNon}, \ref{sec:InfVert}, and \ref{sec:InfHoriz} we introduce infinite-horizon safe models in which the intruder is given no maneuverability, vertical maneuverability, and horizontal maneuverability, respectively. In Sections~\ref{sec:BoundNon} and \ref{sec:BoundVert}, we introduce finite-horizon safe models to act as a stepping stone to the infinite-horizon models in Sections~\ref{sec:SafeNon} and \ref{sec:SafeVert}. Finally, \rref{sec:SafeNon} uses the concept of \emph{safeability} from our previous work~\cite{acasx}, in which the ownship can follow an initial advisory for finite time, and a subsequent advisory forever after. \rref{sec:SafeVert} adds intruder maneuverability to this scenario. 

The models we consider come in three categories: infinite-time models, $\varepsilon$-time models, and safeable models, which increase in complexity. Each category introduces a model that does not grant the intruder any maneuverability to establish intuition, before introducing the model(s) in which the intruder may maneuver.
\emph{The \KeYmaeraX models and proofs of all theorems are online\footnote{All \KeYmaeraX models and proofs are at \url{https://github.com/LS-Lab/KeYmaeraX-projects/tree/master/acasx/acasx-games}}.}

\section{Overview of the ACAS~X Modeling Approach} \label{sec:modelingapproachSec}

To establish intuition for our modeling approach of these flight scenarios, consider a scenario in which an ownship and an intruder are in the same flight space. The intruder at any point in time has the option to change its trajectory within a reasonable bound; the union of all of these possible trajectories at any future point in time describes the unsafe region for the ownship. If at a point in time, the ownship puck overlaps with a possible position that the intruder could be at at that time, we know that such an ownship trajectory is not provably safe because there is a series of ownship and intruder actions which could lead to an NMAC. Therefore, if the ownship is outside of this region then an NMAC cannot possibly occur, and the ownship is safe. 

\begin{figure}[tbhp]
    \centering
    \includegraphics[scale=\imagescale]{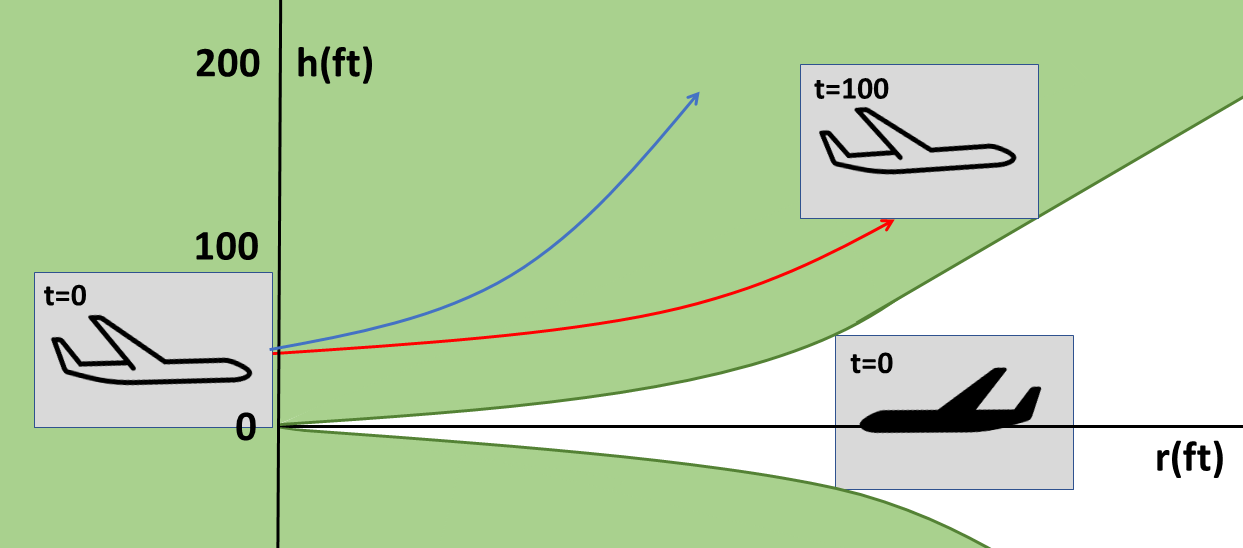}
    \Description[Nominal ownship trajectory and example of a compliant trajectory]{Illustrates the nominal trajectory of an ownship accelerating towards an upsense advisory, with an example of a compliant trajectory. The ownship passes safely above the intruder.}
    \caption{Nominal trajectory (red) within the safe region (green) of an ownship accelerating towards an upsense advisory, with an example of a compliant trajectory (blue).}
    \label{fig:firstnominal}
\end{figure}

Figure~\ref{fig:firstnominal} exemplifies a head-on encounter with the associated safe region for the intruder when the ownship follows a CL1500 advisory per \rref{tab:1}. The coordinate system is fixed at the intruder and centered on the initial position of the ownship. The ownship starts at a relative vertical separation of 0, but a large horizontal separation from the intruder. Upon receiving the CL1500 advisory, it accelerates upwards with acceleration at least $a_{\text{lo}}$, but within the aircraft limits $a_\text{max}$. Once it reaches a vertical velocity of at least 1500\,ft/min, it follows a linear path upwards until clearing the intruder aircraft. The green region is the region of safety which guarantees no NMAC (as long as the ownship follows the advisory), and the red line is the \emph{nominal trajectory} representing minimal compliance with the advisory; the ownship can always choose to accelerate more than $a_{\text{lo}}$ or reach a final upward velocity which is greater than the advisory, and this still qualifies as following the advisory giving uncountably many possible flight trajectories. 
An orthogonal question of equal impact on the safety of the outcome is the sequence of choices of the intruder aircraft.

\subsection{Dynamics} \label{sec:dynamicsSec}
Figure~\ref{fig:IO} shows one encounter between ownship $O$ and intruder $I$. We follow conventions established in the ACAS~X community \cite{MDP}, letting $r = \lVert \textbf{r}\rVert$ be the horizontal distance and $h$ be the vertical separation between the two aircraft. Both positions of the intruder are relative to the ownship. 
	
\begin{figure}[tbhp]
\centering
\begin{subfigure}{.5\textwidth}
  \centering
  \includegraphics[scale=0.4]{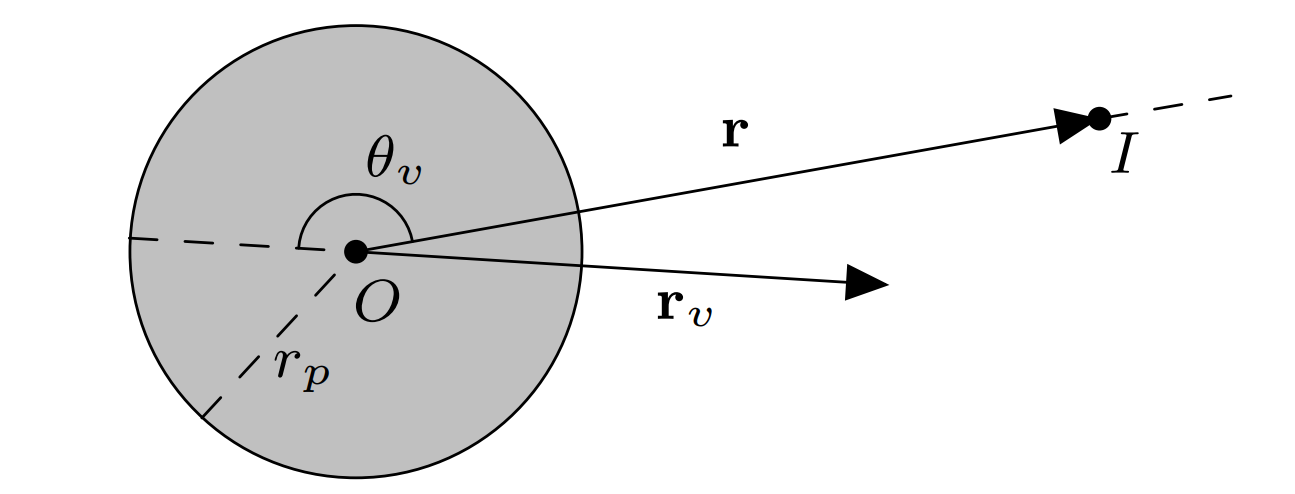}
  \caption{Top view of the encounter} 
  \label{fig:sub1}
\end{subfigure}%
\begin{subfigure}{.5\textwidth}
  \centering
  \includegraphics[scale=0.5]{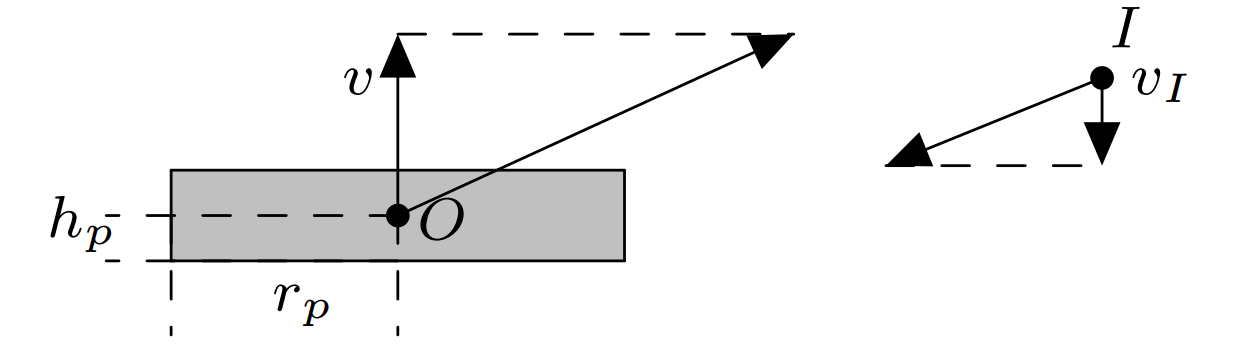}
  \caption{Side view of the encounter}
  \label{fig:sub2}
\end{subfigure}
\Description[An encounter scenario between ownship and intruder, with encasing puck]{An encounter scenario between ownship and intruder. The puck encasing the ownship describes the space around the ownship that may not be entered by intruders.}
\caption{An encounter scenario between ownship $O$ and intruder $I$, with encasing puck shown in gray \cite{acasx}}
\label{fig:IO}
\end{figure}

In this article, we relax assumptions from previous work \cite{acasx2,acasx} in order to give more maneuverability to the intruder. First, we do not necessarily assume that the horizontal rate of closure $r_v$ between the two aircraft is constant. Specifically, in the model in \rref{sec:InfHoriz}, we grant the intruder limited control over this value. This corresponds to the intruder being able to change direction in the horizontal plane during the encounter, represented in \rref{fig:IO} by the $\theta_v$ angle between $r_v$ and $r$. 
	
Second, in the vertical direction, we not only allow the ownship's vertical velocity $v$ to change at any moment, as in previous work, but we also grant the intruder limited control over its own vertical velocity $v_I$ (Sections \ref{sec:InfVert}, \ref{sec:BoundVert}, and \ref{sec:SafeVert}). In all encounters, we assume that the vertical acceleration of the intruder cannot exceed constant $c$ and that of the ownship cannot exceed $a_\text{max}$. Any aircraft will have a rate of vertical acceleration which it cannot exceed due to the physical maneuverability limitations of the aircraft, and it is reasonable to assume for the ownship to have access to this value given the aircraft type of the intruder. 

While these assumptions still limit the possible trajectories of each aircraft about which we will prove safety properties, they are necessary in the modeling and verification process. For instance, while it would be excellent to prove that the ownship can strategically wiggle out of a collision with any aircraft, this is just not possible if the intruder aircraft is strictly more maneuverable than the ownship. Thus, the $c$ constant is necessary to prove meaningful safety properties, even if it limits the types of encounters to which these safety properties apply.

\subsection{Advisories} \label{sec:advisoriesSec}
ACAS~X advisories (except for the Clear-of-Conflict and Multi-Threat Level-Off advisories) have two components: a target velocity $v_{\text{lo}}$ and the direction of the target $w = \pm 1$. For example, the advisory CL1500 specifies that the pilot should achieve a climb rate of at least $1,500$\,ft/min, meaning the target velocity is $1500$ and the direction is upwards ($w=+1$) allowing larger climb rates. For the DNC2000 advisory, the pilot is advised not to climb more than $2,000$\,ft/min. This would make $v_{\text{lo}} = 2000$ and $w = -1$. The $w$ and $v_{\text{lo}}$ values of the ACAS~X advisories are in \rref{tab:1}.

\subsection{Model Overview} \label{sec:modeloverviewSec}
We present a high-level model whose basic structure other models in this article follow. 

\setcounter{modelline}{0}%
\begin{equation}
\label{eq:modelshape}
\begin{aligned}
\text{init} &\,\bigl|\mline{line:shape-init-r} \text{init}(r_p,h_p,w,a_\text{lo},a_\text{max},c) \land R(r, h, v, w,v_{\text{lo}}) \limply\\
\text{advisory} &\,\bigl|\mline{line:shape-adv} \bigl[ \bigl( ((\prandom{w,v_{\text{lo}})} ;~ \ptest{R(r, h, v, w,v_{\text{lo}})} ;~ \text{advisory}:=(w,v_{\text{lo}})) \\
\text{ownship} &\,\bigl|\mline{line:shape-ownship} \phantom{[(} \bigl( \pumod{a_o}{\text{ownship}(\text{advisory})}; \ptest{(-a_{\text{max}} \leq a_o \leq a_{\text{max}})} \bigr)^d \\
\text{intruder} &\,\bigl|\mline{line:shape-intruder} \phantom{[(} \bm{(a_i:=*; \ptest{-c < a_i < c};} \\
\text{motion} &\,\bigl|\mline{line:shape-motion} \phantom{[((} \{ \pevolvein{\D{r}=-r_v,\D{h}=-v,\D{v}=a_o-a_i}{  \text{EDC}(v,v_{\text{lo}},a_o,a_i,a_{\text{lo}})} \} \\
&\,\phantom{\bigl|}\mline{line:shape-innerloopend} \phantom{[(} \bm{)^\ast} \\
\neg\text{NMAC}&\,\bigl|\mline{line:shape-safe} \bigl)^\ast\bigl]\bigl(\lvert r \rvert > r_p \lor \lvert h \rvert > h_p\bigr)
\end{aligned}
\end{equation}

This \dGL formula~\eqref{eq:modelshape} of the shape $P \limply \dbox{\alpha}{Q}$ says that there is a winning strategy for the ownship in the hybrid game $\alpha$ starting in any state satisfying logical formula $P$ to end up in a state satisfying $Q$. The preconditions $P$ ensure both nature-imposed and safety-imposed conditions. Puck radii $r_p$ and $h_p$ must be positive, for instance, as they represent distances, $w=-1 \lor w=1$ since it flags an advisory as either upsense or downsense, and we need relationships between minimum advisory compliant climb rate $a_\text{lo}$, ownship climb rate limits $a_\text{max}$, and intruder climb rate limits $c$. 
We also that the ownship is initially in a safe region $R$ for \emph{some} initial advisory $(w,v_{\text{lo}})$, otherwise we cannot conclude that it will be safe in the future. This is symbolically represented above by the formula $R(r, h, v, w,v_{\text{lo}})$, but this region is both specific to the model being studied and critical to proving safety, and will therefore be developed and explained in great detail in each section. 

The game on lines~\ref{line:shape-adv}--\ref{line:shape-intruder} encodes the sequence of discrete choices made, followed by the evolution of the continuous dynamics on \rref{line:shape-motion}.
Specifically, an advisory is computed ($\prandom{(w,v_\text{lo})}$) and issued on \rref{line:shape-adv} which must satisfy our safe region $\ptest{R(r, h, v,w,v_{\text{lo}})}$, after which the ownship is allowed in \rref{line:shape-ownship} to choose the particular acceleration $a_o$ within the climb rate limits of the aircraft ($-a_{\text{max}}\leq a_o \leq a_{\text{max}}$) that it wants to follow during the encounter. 
This choice of $a_o$ can access the advisory from \rref{line:shape-adv}, but not the specific intruder choice $a_i$ from \rref{line:shape-intruder} after it.
Note that an important switch in coordinates over \rref{tab:1} occurs with respect to $v_\text{lo}$. 
\rref{tab:1} uses $v_\text{lo}$ to refer to the climb rate requested from the ownship pilot, while in all our models $v_\text{lo}$ refers to an advisory in terms of \emph{relative} climb rate; the coordinate transformation to the non-relative advisory is assumed to occur in $\pumod{\text{advisory}}{(w,v_\text{lo})}$ from relative $v_\text{lo}$ and intruder velocity $v_I$ at the time of issuing the advisory.
Then on \rref{line:shape-intruder}, the intruder can change its own control decision $a_i$ within climb rate limits $c$, which we assume in $\text{init}$ to be strict enough so that the ownship can overcome the worst-case intruder maneuvers given its own bounds $a_\text{max}$.
The differential equations of motion combine ownship and intruder acceleration to affect the relative climb rate $v$, and in turn the vertical separation $h$, while the horizontal separation $r$ is affected by the relative horizontal speed $r_v$. 
The differential equations are followed for any duration of time, as long as the evolution domain constraint ($\text{EDC}(v,v_\text{lo},a,a_\text{lo})$) is true.
The evolution domain constraints vary depending on the model and will be discussed in later sections. 
The ${}^\ast$ operator on \rref{line:shape-innerloopend} indicates that the inner loop from \rref{line:shape-intruder} to \rref{line:shape-innerloopend} can be repeated any number of times so that the intruder can change decisions more often than the ownship.
The bold-face intruder choice on \rref{line:shape-intruder} and the inner loop operator on \rref{line:shape-innerloopend} are omitted from models that do not allow the intruder control over its trajectory.

Crucially, the ownship choice on \rref{line:shape-ownship} is contained within the $({}^d)$ operator, which represents the difference in choice between two players so between the ownship and intruder. 
The \dGL formula $P \limply \dbox{\alpha}{Q}$ states that given preconditions $P$, \emph{there is a winning strategy for the ownship} that wins by successfully achieving $Q$ \emph{for all intruder responses} when playing game $\alpha$. 
All choices within the $({}^d)$ game are resolved by our helpful player, so we need only show that there exists some run of this subgame such that for all runs of the game outside the operator, $Q$ is satisfied. Within the context of this model, this means that there need only be one choice of $a_o$ which ultimately allows the ownship to avoid collision.
Crucially, we will prove that for \emph{any} advisory which satisfies our safe region and for \emph{any} set of intruder actions, the pilot \emph{can} strategically pick her acceleration $a_o$ such that an NMAC does not occur.
By contrast, if this choice of $a_o$ were outside the $({}^d)$ operator, the model would be conjecturing that all of the infinitely many choices for ownship acceleration $a_o$ would have to satisfy the postcondition, which it simply does not if the pilot does not pay attention.  

The last ${}^\ast$ operator in \rref{line:shape-safe} is the outer loop around the entire program, which means that the aircraft encounter game can be repeated any number of times. More specifically, the pilot can be given any number of advisories by the aircraft, and our postcondition guarantees that \emph{any sequence} of advisories which satisfy our safe region will guarantees collision freedom at all points in time. 

\subsection{Formalization and Verification Overview}
\label{sec:proofoverview}

The models and theorems in the following sections build upon the general shape \eqref{eq:modelshape} in an incremental fashion. 
An overview of the relationship between the models, definitions of safe regions, and safety theorems is given in \rref{fig:modeloverview}.

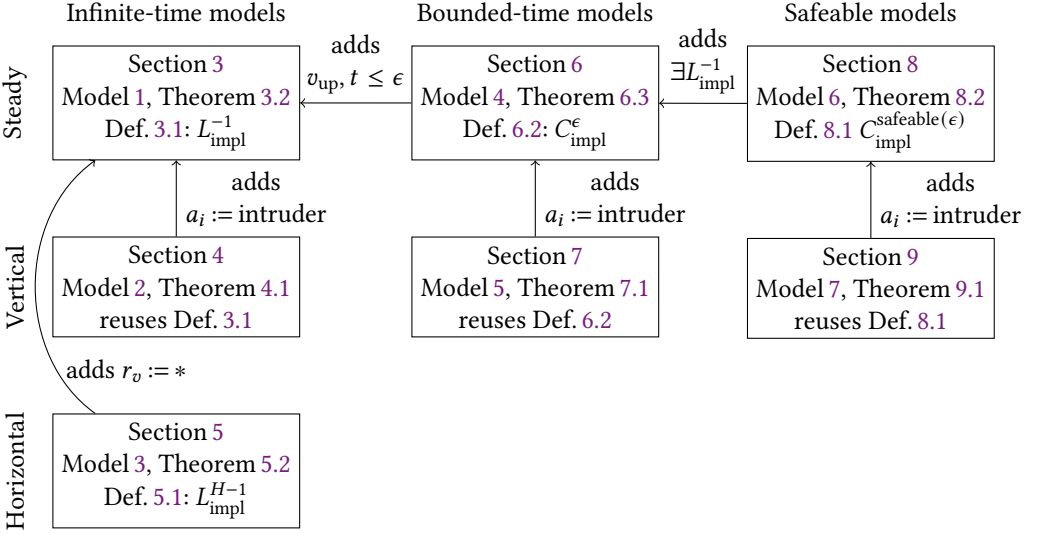
\begin{figure}[htb]
\begin{tikzpicture}[
    modelnode/.style={rectangle, draw=black},
    every text node part/.style={align=center}
]
    \node (infinitetime) {Infinite-time models};    
    \node[modelnode,below=.2cm of infinitetime] (infinitenon) {\rref{sec:InfNon} \\ \rref{model:infinitenon}, \rref{thm:nonmaneuvering} \\ \rref{def:implicitinfinite}: $L^{-1}_\text{impl}$};
    \node[modelnode,below=of infinitenon] (infinitevert) {\rref{sec:InfVert} \\ \rref{model:infinitevert}, \rref{thm:infinitevert} \\ reuses \rref{def:implicitinfinite}};
    \draw[->] (infinitevert) to node[right]{adds \\ $\pumod{a_i}{\text{intruder}}$} (infinitenon);
    \node[modelnode,below=of infinitevert] (infinitehoriz) {\rref{sec:InfHoriz} \\ \rref{model:infinitehoriz}, \rref{thm:implicitinfinitehoriz} \\ \rref{def:implicitinfinitehorizontal}: $L^{H-1}_\text{impl}$};
    \draw[->] (infinitehoriz) to[out=145,in=215] node[pos=0.2,right]{adds $\prandom{r_v}$} (infinitenon);
    \node[left=.5cm of infinitenon,rotate=90,anchor=center] (nonmaneuvering) {Steady};
    \node[left=.5cm of infinitevert,rotate=90,anchor=center] (vertmaneuvering) {Vertical};
    \node[left=.5cm of infinitehoriz,rotate=90,anchor=center] (horizmaneuvering) {Horizontal};
    \node[right=1.5cm of infinitetime] (boundedtime) {Bounded-time models};
    \node[modelnode,below=.2cm of boundedtime] (boundednon) {\rref{sec:BoundNon} \\ \rref{model:boundednon}, \rref{thm:bounded-non} \\ \rref{def:twosidedepsilon}: $C^{\epsilon}_\text{impl}$};
    \draw[->] (boundednon) to node[above]{adds \\ $v_\text{up},t \leq \epsilon$} (infinitenon);
    \node[modelnode,below=of boundednon] (boundedvert) {\rref{sec:BoundVert} \\ \rref{model:boundedvert}, \rref{thm:bounded-vert} \\ reuses \rref{def:twosidedepsilon}};    
    \draw[->] (boundedvert) to node[right]{adds \\ $\pumod{a_i}{\text{intruder}}$} (boundednon);    
    \node[right=1.5cm of boundedtime] (safeable) {Safeable models};
    \node[modelnode,below=.2cm of safeable] (safeablenon) {\rref{sec:SafeNon} \\ \rref{model:safeablenon}, \rref{thm:safeablenon} \\ \rref{def:implicitsafeability} $C^{\text{safeable}(\epsilon)}_\text{impl}$};
    \draw[->] (safeablenon) to node[above]{adds \\ $\exists L^{-1}_\text{impl}$} (boundednon);
    \node[modelnode,below=of safeablenon] (safeablevert) {\rref{sec:SafeVert} \\ \rref{model:safeablevert}, \rref{thm:safeablevert} \\ reuses \rref{def:implicitsafeability}};
    \draw[->] (safeablevert) to node[right]{adds \\ $\pumod{a_i}{\text{intruder}}$} (safeablenon);
\end{tikzpicture}
\caption{Overview of the relationship between models, safe regions, and theorems}
\label{fig:modeloverview}
\end{figure}

\section{Infinite-Time Safety for a Non-Maneuvering Intruder} \label{sec:InfNon}

In this section, we establish a baseline model in which the ownship and intruder are approaching each other at a constant horizontal velocity $r_v$. In future sections, the intruder will have control over its vertical acceleration or the horizontal rate of closure to present a greater challenge to collision avoidance, but for now, the model is kept simple to establish a baseline understanding. 

In this section, we define an advisory to be safe only if it is safe indefinitely, meaning that no further advisory is needed to provide long-term safety of the ownship. This is too restrictive for reasons discussed in \rref{sec:BoundNon}, but it allows for a simple model with which to start our investigation.

\subsection{Model} \label{sec:InfNonModel}

\setcounter{modelline}{0}
\begin{model}[tbhp]
\caption{Infinite-time safety for a non-maneuvering intruder}
\label{model:infinitenon}
\begin{align*}
\text{init}     &\,\bigl|\mline{line:infinitenon-init} \rp \geq0\land\hp>0\land r_v\geq0\land a_{\text{lo}}>0\land(w=-1\lor w=1)\land a_{\text{max}} \geq a_{\text{lo}} \land {} \\
\text{R}        &\,\bigl|\mline{line:infinitenon-r} L_{\text{impl}}^{-1}(r,h,v,w,v_{\text{lo}}) \\
                &\,\phantom{\bigl|}\mline{line:infinitenon-imply}\rightarrow  \\
                &\,\phantom{\bigl|}\mline{line:infinitenon-loopstart}\bigl[\bigl( \\
\text{advisory} &\left|
\begin{aligned}
  &\mline{line:infinitenon-adv-keep} \quad \bigl(\phantom{\cup}\phantom{\bigl(} \ptest{\ltrue} \\
  &\mline{line:infinitenon-adv-new} \quad \phantom{\bigl(}\cup (\pchoice{\pumod{w}{1}}{\pumod{w}{-1}});\pumod{v_{\text{lo}}}{\ast}; \ptest{L_{\text{impl}}^{-1}}(r,h,v,w,v_{\text{lo}}); \pumod{\text{advisory}}{(w,v_{\text{lo}})}\bigr); 
 \end{aligned}
 \right.\\
\text{ownship}  &\,\bigl|\mline{line:infinitenon-ownship} \phantom{\bigl(}\quad\pdual{\bigl(\pumod{a_o}{\ast}; \ptest{(-a_{\text{max}} \leq a_o \leq a_{\text{max}})}\bigr)}; \\
\text{motion}   &\left|
\begin{aligned}
    &\mline{line:infinitenon-motion-ode} \phantom{\bigl(} \quad\{\D{r}=-r_v \syssep \D{h}=-v\syssep\D{v}=a_o\}
\end{aligned}
\right.\\
\neg\text{NMAC} &\,\Bigl|\mline{line:infinitenon-safe} \phantom{\Bigl[\,}\bigr)^\ast\bigr]\bigl(\lvert r \rvert > \rp \lor \lvert h \rvert > \hp\bigr)
\end{align*}
\end{model}

The formula $L_{\text{impl}}^{-1}(r, h, v, w, v_{\text{lo}})$ on \rref{line:infinitenon-r} is the safe region of this model: an ownship originally separated from an intruder by $r$ horizontally and $h$ vertically will avoid collision given the advisory $(w,v_{\text{lo}})$. This crucial region will be developed and explained in \rref{sec:InfNonRegions}.
The postcondition on \rref{line:infinitenon-safe} expresses the desire that there must always be a separation of the puck distance between the two aircraft in either the horizontal or vertical direction, so no NMAC ever occurs.

Lines \ref{line:infinitenon-adv-keep}--\ref{line:infinitenon-adv-new} encode the advisory. The nondeterministic operator $\cup$ encodes that the pilot has two options: either to continue with her current advisory, in which case only the $(\ptest{\ltrue})$ condition must be satisfied, or follow a new advisory. The $?$ operator discards any runs of the system in which the condition after the $?$ is false, so $(\ptest{\ltrue})$ is always trivially satisfied. However, this $(\ptest{\ltrue})$ condition is important to ensure that the system always has a valid choice for an advisory (i.e. keep the previous advisory), and will not get stuck without any safe advisories. 

The requirement on \rref{line:infinitenon-adv-new} is that the choice of target velocity and up $(w=1)$ or down $(w=-1)$ advisory will keep the ownship within the safe region indefinitely, as encoded by the $?(L_{\text{impl}}^{-1}(r, h, v, w, v_{\text{lo}}))$ condition, given the minimum acceleration $a_{\text{lo}}$ from \rref{sec:ACASXSec}. Without this assumption of a minimum acceleration, there is no guarantee in how quickly an aircraft would reach its target velocity, and therefore no safety guarantees.

Line~\ref{line:infinitenon-ownship} expresses the pilot's choice of her own vertical acceleration $(\prandom{a_o})$. While this choice at first seems arbitrary in the context of the model, the proof will need to make strategic choices for $a_o$ within the climb rate limits $a_{\text{max}}$ of the aircraft to pass the subsequent test and avoid NMACs.
Crucially, the choice of $a_o$ and test in \rref{line:infinitenon-ownship} are within a $({}^d)$ operator since the choice of $a_o$ and the responsibility of staying within the climb rate limits of the ownship are up to the ownship pilot.

The two aircraft then follow the differential equations on \rref{line:infinitenon-motion-ode}. The differential equations express that $r$, $v$, and $h$ of the ownship evolve according to $r_v$, the rate of horizontal closure, as well as the acceleration $a_o$ chosen by the pilot based on the advisory issued by the system. 

\subsection{Implicit Formulation of the Safe Region} \label{sec:InfNonRegions}
As previously stated, the safety of this model hinges on proving that an aircraft following an advisory from \rref{line:infinitenon-adv-new} will stay within the safe region throughout the encounter with an intruder. This begs the question of what regions of flight guarantee safety for the ownship. Just like the previous work, we represent this region with the formula $L_{\text{impl}}^{-1}(r, h, v, w, v_{\text{lo}})$. This region is fixed at the intruder with its origin at the initial position of the ownship. We will explain one case in detail and provide a generalization of this region after.

\noindent\paragraph{First case:}
Consider the case of an upsense advisory $w = +1$, where the ownship has not yet reached the target velocity $(v \leq v_{\text{lo}})$. 
We use the concept of a nominal trajectory \cite{acasx} (denoted by $\mathcal{N}(t)$), an example of which is shown in \rref{fig:firstnominal}
in red. In this figure, the ownship follows one possible trajectory adhering to all the requirements of ACAS~X. This is just one of the uncountably infinitely many safe trajectories of the ownship. Upon receiving an advisory from ACAS~X, the ownship begins climbing at an acceleration of $a_{\text{lo}}$ and continues climbing along a parabola until it reaches the advised velocity $v_{\text{lo}}$. It then stops climbing and continues at the vertical velocity $v_{\text{lo}}$ in a straight line. Through integration of the equations of motion from \rref{model:infinitenon} \rref{line:infinitenon-motion-ode}, we can get the coordinates $(r_n, h_n)$ of the ownship along this nominal trajectory as a function of time: 
\[
  (r_n,h_n) =
  \begin{cases}
        \left(r_vt, \frac{a_{\text{lo}}}2t^2+vt \right) & \text{if $0 \leq t < \frac{v_{\text{lo}}-v}{a_{\text{lo}}}$} \\
        \left(r_vt,v_{\text{lo}}t - \frac{(v_{\text{lo}}-v)^2}{2a_{\text{lo}}} \right) & \text{if $\frac{v_{\text{lo}}-v}{a_{\text{lo}}} \leq t$} \\
  \end{cases}
\]

While these equations describe the minimally-compliant nominal trajectory, the ownship pilot's choice of $a_o$ in \rref{line:infinitenon-ownship} can be different than $a_{\text{lo}}$ (and may exceed it while necessary). The actual coordinates of the ownship could be anywhere above this safe nominal trajectory, creating a region in which the ownship is guaranteed to be. Thus, the ownship is safe if it is separated horizontally from the nominal trajectory by at least the puck width $(\vert r-r_n\vert > r_p)$, or it is above the nominal trajectory by at least puck height $(h_n - h > h_p)$, that is: 
$$\forall t\, \forall r_n\,\forall h_n\, ((r_n,h_n) \in \mathcal{N}(t) \rightarrow \vert r - r_n \vert > r_p \lor h_n-h >h_p)$$
This will be referred to as the implicit formulation of the safe region. It is an implicitly defined region because it uses quantifiers as opposed to explicit inequalities to define the nominal trajectory. 

\noindent\paragraph{Generalization:}
The same reasoning applies to the $w=-1$ case where the pilot is told to descend to avoid collision as well as the $v > v_{\text{lo}}$ case, where the ownship has already achieved the target velocity and is now following the straight-line trajectory. 

Going back to \rref{model:infinitenon}, the test $\ptest{L_{\text{impl}}^{-1}(r,h,v,w,v_{\mathit{lo}})}$ on \rref{line:infinitenon-adv-new} guarantees that following the nominal trajectory keeps the ownship safe, and thus we know that if the pilot accelerates at least as fast as the minimum acceleration $a_{\text{lo}}$ or has already reached the target velocity, then she is above the nominal trajectory and is therefore safe as well. The test also allows the ACAS~X system the flexibility to give \emph{any} advisory which results in a safe nominal trajectory, and the pilot the flexibility to choose arbitrary accelerations which keep the plane in the implicit region. This safe region is used to prove the safety postcondition of \rref{model:infinitenon}, making it sufficient to reason about this region to guarantee the safety of the ownship. The implicit formulation $L_{\text{impl}}^{-1}(r, h, v, w, v_{\text{lo}})$ follows \cite[Fig. 3]{acasx}, is listed in \rref{def:implicitinfinite}, and used in \rref{thm:nonmaneuvering}, which has been verified using \KeYmaeraX.

\begin{definition}[Implicit Infinite-Time Safe Region]
\label{def:implicitinfinite}
\begin{align*}
    T_{\text{lo}}(v,w,v_{\text{lo}}) &\equiv \frac{\text{max}(0,w(v_{\text{lo}}-v))}{a_{\text{lo}}}\\
    A_{\text{lo}}(v,w,v_{\text{lo}},h_n,t) &\equiv \biggl( 0 \leq t < T_{\text{lo}}(v,w,v_{\text{lo}}) \land h_n = \frac{wa_{\text{lo}}}2t^2 + vt\biggr)\\
    & \phantom{\equiv} \lor \biggl( t \geq T_{\text{lo}}(v,w,v_{\text{lo}}) \land h_n = v_{\text{lo}}t - \frac{w\text{max}(0,w(v_{\text{lo}}-v))^2}{2a_{\text{lo}}} \biggr)   \\
    L_{\text{impl}}^{-1}(r,h,v,w,v_{\text{lo}}) &\equiv \forall t \, \forall r_n \, \forall h_n \left( r_n = r_vt \land A_{\mathit{lo}}(v,w,v_{\text{lo}},h_n,t) \rightarrow (\vert r-r_n\vert >r_p \lor w(h_n-h) > h_p )\right)
\end{align*}
\end{definition}

\begin{theorem}[Non-maneuvering intruder: correctness of implicit safe regions] 
\label{thm:nonmaneuvering}
The \dGL formula given in \rref{model:infinitenon} is valid. That is as long as the advisories followed obey formula $L_{\text{impl}}^{-1}(r,h,v,w,v_{\text{lo}})$ from \rref{def:implicitinfinite} the winning strategy will avoid NMAC.
\end{theorem}

\begin{proof}
The \KeYmaeraX proof develops a winning strategy for choosing ownship control $a_o$:
\begin{equation*}
a_o = \begin{cases}
      wa_{\text{lo}} & \text{if}~ wv < wv_{\text{lo}}\\
      0 & \text{if}~ wv \geq wv_{\text{lo}}
\end{cases}
\end{equation*}
We use minimal vertical acceleration $a_\text{lo}$ in direction $w$ to adjust the climb rate towards the advisory.
When the advisory is met (when $wv \geq wv_{\text{lo}}$), we keep the climb rate steady by picking $a_o=0$.
\end{proof}

\section{Infinite-Time Safety for a Vertically-Maneuvering Intruder} \label{sec:InfVert}
Now that we have established a baseline model for this encounter which includes the important $({}^d)$ duality operator for player selection in hybrid games, we can continue on to more expressive models. Our next model accounts for vertical intruder maneuvers, and the model expresses the notion that the ownship should always have a way to overcome any reasonable intruder action. In the example shown in \rref{fig:intruderaccel}, even though the intruder begins accelerating in the same direction as the pilot, the pilot can still overcome the intruder action and navigate to safety.

\begin{figure}[tbhp]
    \centering
    \includegraphics[scale=\imagescale]{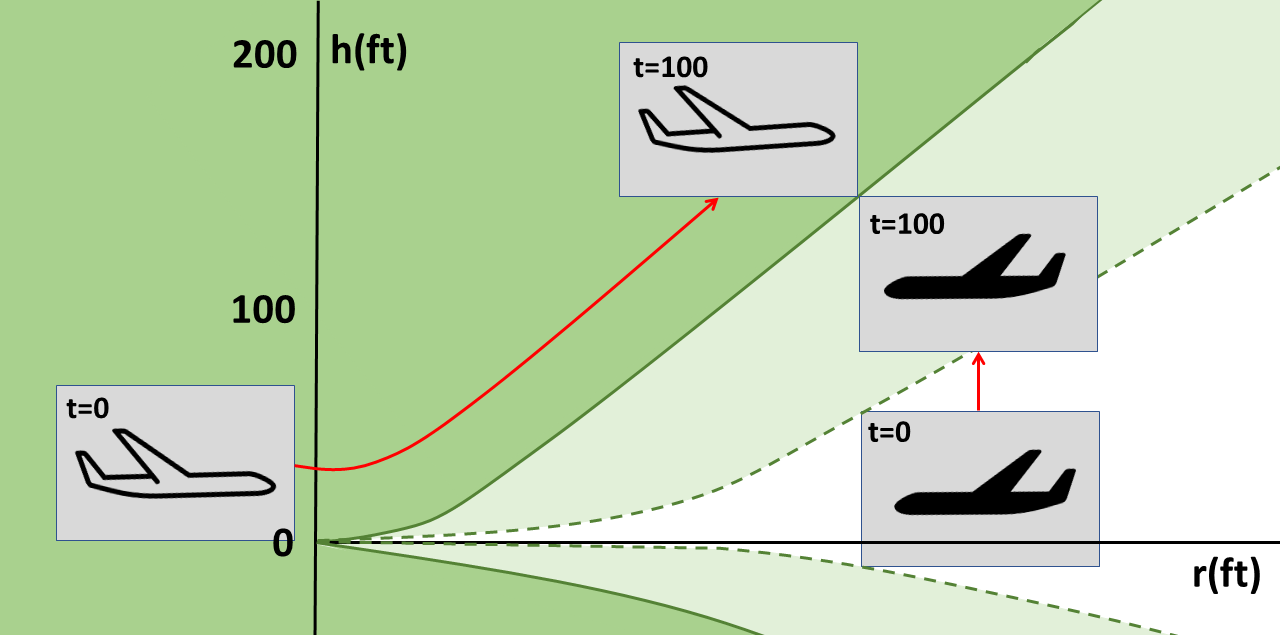}
    \Description[Nominal trajectory for an encounter with a vertically maneuvering intruder.]{Illustrates the nominal trajectory and safe region of an ownship accelerating towards an upsense advisory, where the intruder is accelerating in the same direction. The safe region for an encounter with a maneuvering intruder is smaller than the safe region for an encounter with a non-maneuvering intruder.}
    \caption{Nominal trajectory (red) and safe region (dark green) of an ownship accelerating towards an upsense advisory, where the intruder is accelerating in the same direction. The larger region that would be safe \emph{without} intruder maneuverability is shown in light green. }
    \label{fig:intruderaccel}
\end{figure}

\subsection{Model} \label{sec:InfVertModel}

\setcounter{modelline}{0}
\begin{model}[tbhp]
\caption{Infinite-time safety for a vertically-maneuvering intruder}
\label{model:infinitevert}
\begin{align*} 
\text{init}     &\,\bigl|\mline{line:infinitevert-init} \rp \geq0\land\hp>0\land r_v\geq0\land a_{\text{lo}}>0\land \bm{c > 0} \,\land(w=-1\lor w=1)\land \bm{a_\textbf{max} \geq a_\textbf{lo}+c} \\
\text{R}        &\,\bigl|\mline{line:infinitevert-r} L_{\text{impl}}^{-1}(r,h,v,w,v_{\text{lo}}) \\
                &\,\phantom{\bigl|}\mline{line:infinitevert-imply} \limply  \\
                &\,\phantom{\bigl|}\mline{line:infinitevert-loopstart} \bigl[\bigl( \\
\text{advisory} &\left|
\begin{aligned}
  &\mline{line:infinitevert-adv-keep}\phantom{\bigl(}\quad\bigl(\phantom{\cup}\phantom{\bigl(} \ptest{\ltrue} \\
  &\mline{line:infinitevert-adv-new}\phantom{\bigl(}\quad\phantom{\bigl(}\cup (\pchoice{\pumod{w}{1}}{\pumod{w}{-1}}); \pumod{v_{\text{lo}}}{\ast}; \ptest{L_{\text{impl}}^{-1}(r,h,v,w,v_{\text{lo}})}; \pumod{\text{advisory}}{(w,v_{\text{lo}})}\bigr); 
 \end{aligned}
 \right.\\
\text{ownship}  &\,\bigl|\mline{line:infinitevert-ownship}\phantom{\bigl(}\quad\pdual{\bigl(\pumod{a_o}{\ast}; \ptest{(-a_\text{max} \leq a_o \leq a_\text{max}})\bigr)}; \\
\text{intruder}  &\,\bigl|\mline{line:infinitevert-intruder}\phantom{\bigl(}\quad\bm{\bigl(
a_{i} := *; \,?(-c < a_{i} < c);} \\
\text{motion}   &\left|
\begin{aligned}
    &\mline{line:infinitevert-motion-ode}\phantom{\bigl(}\phantom{(}\quad\{\D{r}=-r_v \syssep \D{h}=-v\syssep\bm{v'=a_o-a_i}\}
\end{aligned}
\right.\\
&\,\phantom{\bigl|}\mline{line:infinitevert-innerloopend}\phantom{\bigl(}\quad\bm{\bigr)^\ast}  \\
\neg\text{NMAC} &\,\Bigl|\mline{line:infinitevert-safe} \phantom{\Bigl[\,}\bigr)^\ast\bigr]\bigl(\lvert r \rvert > \rp \lor \lvert h \rvert > \hp\bigr)
\end{align*}
\end{model}

In this section, the setup is very similar to the last. The intruder and ownship still approach one another with a constant horizontal rate of closure $r_v$, and we have the same relative coordinate system with $r$ for relative horizontal distance between the aircraft and $h$ for relative vertical distance. 

The difference in this model, highlighted in \textbf{bold}, is that we add the variable $a_i$ representing the intruder's vertical acceleration, as well as a constant $c$ which represents the maximum magnitude of intruder acceleration, discussed previously in \rref{sec:dynamicsSec}. On \rref{line:infinitevert-intruder}, the intruder is now able to nondeterministically change its acceleration $a_i$ within $-c$ and $c$. The differential equations on \rref{line:infinitevert-motion-ode} have also been modified to incorporate the new dynamics of the intruder. The rate of vertical closure between the two aircraft $v$ now evolves at a rate of $(a_o-a_i)$ to reflect that both the intruder and ownship acceleration affect the rate of vertical closure. We choose to treat the rate of vertical closure $v$ as a relative rate as opposed to creating two separate absolute vertical velocities in order to minimize the number of variables needed in this model even if $a_o$ and $a_i$ are independent. 

We add a loop around the intruder choice of acceleration $a_i$ and the dynamics on lines \ref{line:infinitevert-intruder}--\ref{line:infinitevert-innerloopend}. The loop contains the intruder choice of acceleration, but not the ownship's choice, meaning the intruder could change its acceleration \emph{any} number of times during an interaction without the ownship being able to react. Recall that ACAS~X tracks the position and velocity of the intruder, but not its acceleration \cite{FAA1}, so we cannot assume that the ownship has constant knowledge of the intruder's acceleration. Therefore, the ownship must first make a choice in acceleration and stick with it while the intruder is more powerful and allowed to change its acceleration arbitrarily often.

This model comes with new strategic insights that are developed in the proof.
The strategy chooses a \emph{relative} rate of vertical separation (not the absolute ownship velocity) as the target velocity, or that the ownship is accelerating at least at the minimum velocity $a_{\text{lo}}$ plus the maximum intruder acceleration $c$ as a strategy to overcome any possible intruder action. Since the ownship has access to the intruder's velocity and position, it can reasonably monitor the rate of vertical separation. However, since the ownship does not have continuous access to intruder acceleration, it must choose an acceleration which can withstand changes to the intruder's acceleration throughout the encounter. Therefore, if the ownship compensates by accelerating at least $a_{\text{lo}}+c$ then the intruder will not be able to catch the ownship during the encounter, even in the worst cases of the intruder accelerating at $c$ or $-c$, because the relative acceleration difference will be at least $a_{\text{lo}}$. 

\begin{figure}[tbhp]
    \centering
    \includegraphics[scale=\imagescale]{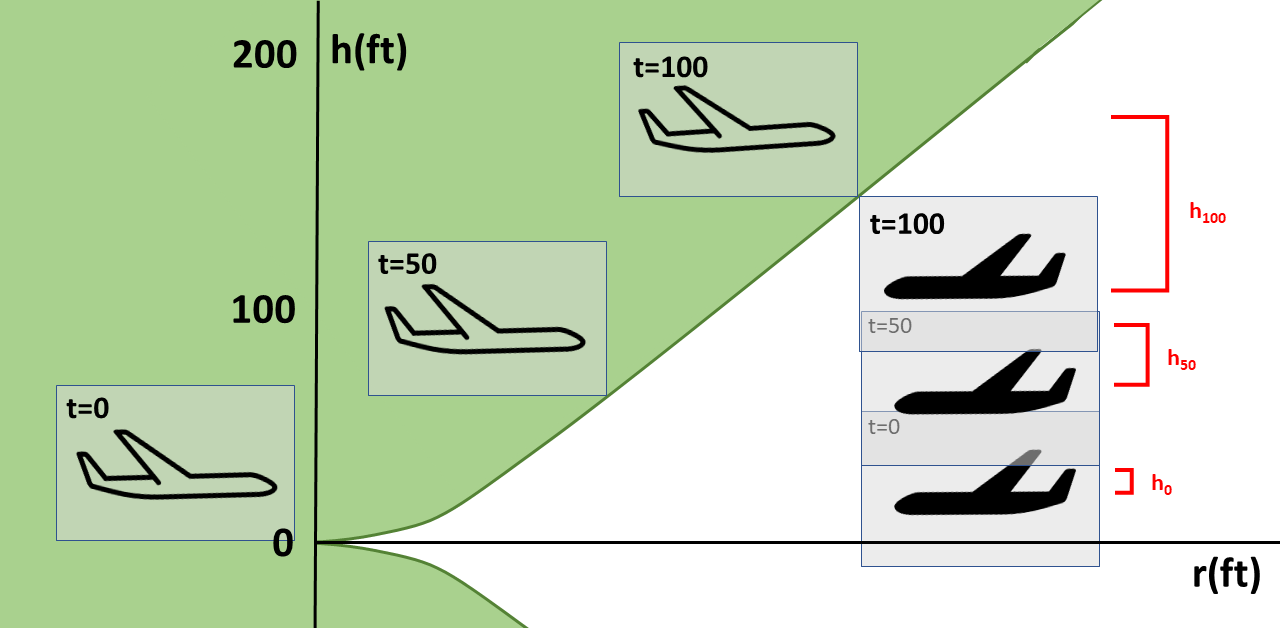}
    \Description[Absolute altitudes and relative separation for a vertically maneuvering intruder.]{Absolute altitudes and relative separation in an encounter between the ownship and a vertically-maneuvering intruder.}
    \caption{An encounter between ownship and vertically-maneuvering intruder showing absolute altitudes of each aircraft and relative separation in red.}
    \label{fig:aloplusc}
\end{figure}

\begin{figure}[tbhp]
    \centering
    \includegraphics[scale=\imagescale]{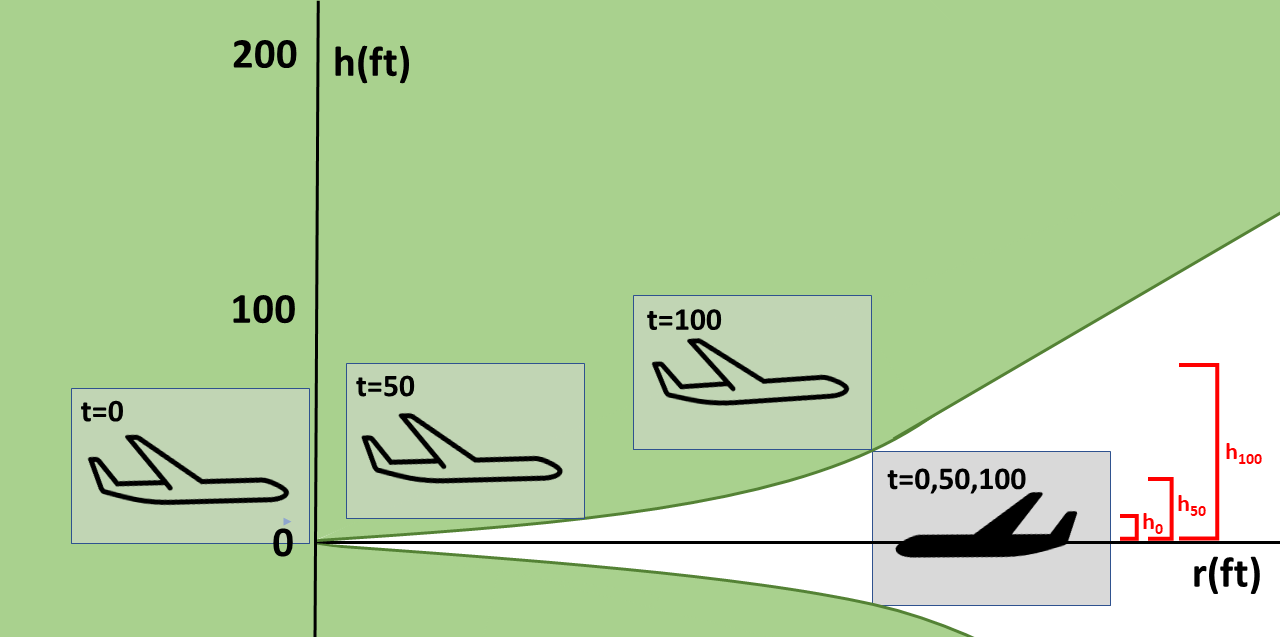}
    \Description[Absolute altitudes and relative separation for a non-maneuvering intruder.]{Absolute altitudes and relative separation in an encounter between the ownship and a non-maneuvering intruder.}
    \caption{An encounter between ownship and non-maneuvering intruder showing absolute altitudes of each aircraft and relative separation in red.}
    \label{fig:alo}
\end{figure}

Figure~\ref{fig:aloplusc} shows an ownship following this new strategy from \rref{thm:infinitevert} with an intruder accelerating vertically at maximum climb rate $c$, and \rref{fig:alo} shows the ownship following the strategy from \rref{thm:nonmaneuvering} with a non-maneuvering intruder. 
Note that the relative separation in the two figures evolves in the same way, since the ownship in \rref{fig:aloplusc} compensates the intruder maneuver by choosing a climb rate that accounts for the intruder maximum climb rate $c$.

\subsection{Implicit Formulation of the Safe Region} \label{sec:InfVertRegion}
Despite the new collision avoidance strategy, we do not need to make any changes to the implicit safe region formulation. This is because the intruder dynamics are hidden by the variables $h$ and $v$ which represent \emph{relative} vertical separation and rate of vertical closure. When the ownship strategically accelerates upwards with at least acceleration $a_{\text{lo}}+c$ until reaching a vertical rate of separation of at least $v_{\text{lo}}$, since the intruder cannot accelerate more than $c$, the relative rate of vertical acceleration is still at least $a_{\text{lo}}$. Therefore, our previous safe region $L_{\text{impl}}^{-1}$ from  \rref{def:implicitinfinite} still applies with the coordinate system still fixed at the intruder. 

\begin{theorem}[Vertically maneuvering intruder: Correctness of implicit safe regions] \label{thm:infinitevert}
The \dGL formula given in \rref{model:infinitevert} is valid. That is as long as the advisories followed obey formula $L_{\text{impl}}^{-1}(r,h,v,w,v_\text{lo})$ from \rref{def:implicitinfinite}, the winning strategy will avoid NMAC.
\end{theorem}

\begin{proof}
The \KeYmaeraX proof develops a winning strategy for choosing ownship control $a_o$:
\begin{equation*}
a_o = \begin{cases}
      w(a_{\text{lo}}+c) & \text{if}~ wv < wv_{\text{lo}}\\
      wc & \text{if}~ wv \geq wv_{\text{lo}}
\end{cases}
\end{equation*}
We compensate intruder maneuvers with increased minimal vertical acceleration $a_{\text{lo}}+c$ in direction $w$ to adjust the climb rate towards the advisory.
When the ownship follows the issued advisory (when $wv \geq wv_{\text{lo}}$), we pick $a_o=wc$ to compensate intruder maneuvers and maintain following the advisory.
The safe region $L_{\text{impl}}^{-1}(r,h,v,w,v_\text{lo})$ serves as a loop invariant of the outer loop of \rref{model:infinitevert}. 
Since the ownship control choice $a_o=wc$ in the winning strategy is only safe when the ownship already follows the issued advisory, in this case we additionally show that following the advisory ($wv \geq wv_\text{lo}$) is an invariant of the inner loop.
\end{proof}

\section{Infinite-Time Safety for a Horizontally-Maneuvering Intruder} \label{sec:InfHoriz}
Where \rref{model:infinitevert} expands on \rref{model:infinitenon} by granting the intruder limited control over its \emph{vertical} velocity, \rref{model:infinitehoriz} in this section grants the intruder limited control over the \emph{horizontal} rate of closure. 
The intruder can increase its ground velocity towards the ownship, as well as change the $\theta_v$ angle shown in \rref{fig:IO} in order to alter the rate of closure between the two aircraft. 

\subsection{Model} \label{sec:InfHorizModel}

\setcounter{modelline}{0}
\begin{model}[tbhp]
\caption{Infinite-time safety for a horizontally-maneuvering intruder}
\label{model:infinitehoriz}
\begin{align*} 
\text{init}     &\,\bigl|\mline{line:infinitehoriz-init} \rp \geq0\land\hp>0\land r_v\geq0\land a_{\text{lo}}>0\land \bm{v_\textbf{max} > 0} \,\land(w=-1\lor w=1)\land {} \\
\text{R}        &\,\bigl|\mline{line:infinitehoriz-r}  \bm{L_{\textbf{impl}}^{H-1}(r,h,v,w,v_{\textbf{lo}})} \\
                &\,\phantom{\bigl|}\mline{line:infinitehoriz-imply} \limply  \\
                &\,\phantom{\bigl|}\mline{line:infinitehoriz-loopstart} \bigl[\bigl( \\
\text{advisory} &\left|
\begin{aligned}
  &\mline{line:infinitehoriz-adv-keep}\phantom{\bigl(}\quad\bigl(\phantom{\cup}\phantom{\bigl(} \ptest{\ltrue} \\
  &\mline{line:infinitehoriz-adv-new}\phantom{\bigl(}\quad\phantom{\bigl(}\cup (\pchoice{\pumod{w}{1}}{\pumod{w}{-1}});\prandom{v_{\text{lo}}}; \bm{\ptest{L_{\textbf{impl}}^{H-1}(r,h,v,w,v_{\textbf{lo}})}}; \pumod{\text{advisory}}{(w,v_{\text{lo}})}\bigr); 
 \end{aligned}
 \right.\\
\text{ownship}  &\,\bigl|\mline{line:infinitehoriz-ownship}\phantom{\bigl(}\quad\pdual{(\prandom{a_o}; \ptest{(-a_\text{max} \leq a_o \leq a_\text{max})})}; \\
\text{intruder}  &\,\bigl|\mline{line:infinitehoriz-intruder}\phantom{\bigl(}\quad
\bigl(\bm{r_v := *; \,?(0 \leq r_v \leq v_\textbf{max});} \\
\text{motion}   &\left|
\begin{aligned}
    &\mline{line:infinitehoriz-motion-ode}\phantom{\bigl(}\phantom{(}\quad\{\D{r}=-r_v \syssep \D{h}=-v\syssep \D{v}=a_o \} 
\end{aligned}
\right.\\
                &\,\phantom{\bigl|}\mline{line:infinitehoriz-innerloopend}\phantom{\bigl(}\quad\bigr)^\ast  \\
\neg\text{NMAC} &\,\Bigl|\mline{line:infinitehoriz-safe} \phantom{\Bigl[\,}\bigr)^\ast\bigr](\lvert r \rvert > \rp \lor \lvert h \rvert > \hp)
\end{align*}
\end{model}

In \rref{model:infinitehoriz}, we give control of the horizontal rate of closure $r_v$ to the intruder.
We do not assume horizontal maneuverability for the ownship, but the intruder has sole control over the horizontal rate of closure, and can therefore nondeterministically choose $r_v$ on \rref{line:infinitehoriz-intruder}. 
This requires the new safe region $L_{\text{impl}}^{H-1}$ on \rref{line:infinitehoriz-adv-new}. Just as in the previous section, we must assume some upper limit on the intruder maneuverability; therefore, we introduce the constant $v_\text{max}$ on \rref{line:infinitehoriz-intruder} to represent the maximum possible horizontal rate of closure between the two aircraft. For simplicity, we also assume that the intruder's maneuvers cannot invert the rate of closure, meaning the intruder cannot fully turn around during the encounter $(\theta_v \in [90^{\circ}, 270^{\circ}])$, and $r_v$ must be non-negative. 
This assumption ensures that the intruder is not fully adversarial and will not ascend or descend along a helix, or turn around to chase the ownship.

Besides these changes, the dynamics mirror that of \rref{model:infinitenon}. 
Again, the loop around the intruder choice and the dynamics allows the intruder to change the rate of closure as many times as it pleases during the encounter, but the ownship must stick with its initial choice of acceleration.

\subsection{Implicit Formulation of the Safe Region} \label{sec:InfHorizRegions}
The safe region of this model must take into consideration the variable rate of closure $r_v$, which could vary anywhere from $0$ to $v_\text{max}$. As such, we no longer have a single nominal trajectory, but infinitely many nominal trajectories $\mathcal{N}_{r_v}(t)$ which are a function of the horizontal rate of closure $r_v$ chosen by the intruder. Each nominal trajectory $\mathcal{N}(r_v)$ is the same as in Sections \ref{sec:InfNon} and \ref{sec:InfVert}:
\[
  (r_n,h_n) =
  \begin{cases}
        \left(r_vt, \frac{a_{\text{lo}}}2t^2+vt \right) & \text{if $0 \leq t < \frac{v_{\text{lo}}-v}{a_{\text{lo}}}$} \\
        \left(r_vt,v_{\text{lo}}t - \frac{(v_{\text{lo}}-v)^2}{2a_{\text{lo}}} \right) & \text{if $\frac{v_{\text{lo}}-v}{a_{\text{lo}}} \leq t$} \\
  \end{cases}
\]

However, in order for our region to be safe, we must know that \emph{each} nominal trajectory whose $r_v$ is within $[0,v_\text{max}]$ keeps a safe distance from the intruder at all future points in time. This motivates universally quantifying over all possible choices of $r_v$ in our new region:  
\[\forall t \forall r_v \forall r_n \forall h_n (r_v \in [0,v_\text{max}] \land (r_n,h_n) \in \mathcal{N}_{r_v}(t) \rightarrow \vert r -r_n \vert > r_p \lor h_n - h > h_p)\]
If this region holds for some advisory $v_{\text{lo}}$, we know that for any possible horizontal rate of closure chosen by the intruder, if the ownship obeys the advisory then it will avoid an NMAC. The implicit formulation $L_{\text{impl}}^{H-1}(r, h, v, w, v_{\text{lo}})$ is shown in \rref{def:implicitinfinitehorizontal} and used in \rref{thm:implicitinfinitehoriz}, which has been verified to be safe using \KeYmaeraX.

\begin{definition}[Lower-bounded, infinite time safe region with horizontally-maneuvering intruder]
\label{def:implicitinfinitehorizontal}
\begin{align*}
    L_{\text{impl}}^{H-1}(r,h,v,w,v_\text{lo}) &\equiv \forall t \, \forall r_v \, \forall r_n \, \forall h_n \bigl( r_v \in [0,\text{max}_v] \land  r_n = r_vt \land A_{\text{lo}}(v,w,v_\text{lo},h_n,t) \\ 
    & \phantom{\equiv \forall t \, \forall r_v \, \forall r_n \, \forall h_n ( r_v \in [0,\text{max}_v] \land} \limply (\vert r-r_n\vert >r_p \lor w(h_n-h) > h_p )\bigr)
\end{align*}
with $A_{\text{lo}}(v,w,v_\text{lo},h_n,t)$ and $T_{\text{lo}}(v,w,v_\text{lo})$ as per \rref{def:implicitinfinite}.
\end{definition}

\begin{theorem}[Horizontally-maneuvering intruder: correctness of implicit safe regions]
\label{thm:implicitinfinitehoriz}
The \dGL formula given in \rref{model:infinitehoriz} is valid. That is as long as the advisories followed obey formula $L_{\mathrm{impl}}^{H-1}(r,h,v,w,v_\mathrm{lo})$ of \rref{def:implicitinfinitehorizontal}, the winning strategy will avoid NMAC.
\end{theorem}

\begin{proof}
The \KeYmaeraX proof develops a winning strategy for choosing ownship control $a_o$:
\begin{equation*}
    a_o=
    \begin{cases}
        wa_\text{lo} & \text{if}~wv < wv_\text{lo} \\
        0 & \text{if}~wv_\text{lo} \leq wv
    \end{cases}
\end{equation*}
We use minimal vertical acceleration $a_\text{lo}$ in direction $w$ to adjust the climb rate towards the advisory.
When the advisory is met (when $wv \geq wv_{\text{lo}}$), we keep the climb rate steady by picking $a_o=0$.
\end{proof}

\section{Bounded-Time Safety for a Non-Maneuvering Intruder} \label{sec:BoundNon}

Up to this point, the models described in this article have all conveyed the notion that advisories issued by the ACAS~X system must be safe \emph{indefinitely}; that is, while the system \emph{can} issue new advisories at any point, the aircraft must avoid collision \emph{whether or not} the system issues a new advisory to pass our rigorous notion of safety. In the context of the above models, this need comes from the lack of a time bound on the differential equations, so \emph{every} duration of the differential equation must guarantee safety for our model to be safe, even without another advisory change.

In a realistic scenario, this idea is limiting because it typically forces the system to give fairly strong advisories too early on and too frequently. While this is acceptable from a safety-perspective, it is unacceptable from a pilot-perspective. A system which gives overly frequent advisories can be distracting to the pilot and may lead the pilot to ignore the system warnings all together. Thus, if an encounter is not immediately threatening, the ACAS~X implementation will actually issue COC or a preventative advisory like DNC or DND, and will later issue a more disruptive advisory as the threat level increases. In cases like this, the preliminary advisory may not be safe indefinitely, but safety can later be restored with a subsequent advisory (as will be the case in \rref{fig:safeable}). 

We represent this idea through the concept of safeability from previous ACAS~X work \cite{acasx}.

\begin{definition}[Safeable] 
We say that an advisory is \emph{safeable} if and only if it is safe or can still be made safe in the future, if necessary, via subsequent advisories.
\end{definition}

In other words, we need only ensure the safety of an advisory for a few seconds as long as \emph{some} followup advisory exists that will keep us safe forever. Once again, the safeable region of a given advisory is always a superset of its safe region, so safeability is a more robust definition.

To develop the safeability models, we assume that the ownship and intruder are approaching at a constant horizontal rate of closure $r_v$, and $r$ and $h$ represent the relative horizontal and vertical separation, respectively. A subsequent advisory is called a \emph{reversal} when the sign of $w$ reverses from the initial advisory. If $w$ does not change, the subsequent advisory is a \emph{strengthening} or \emph{weakening}.

\begin{figure}[b!]
    \centering
    \includegraphics[scale=\imagescale]{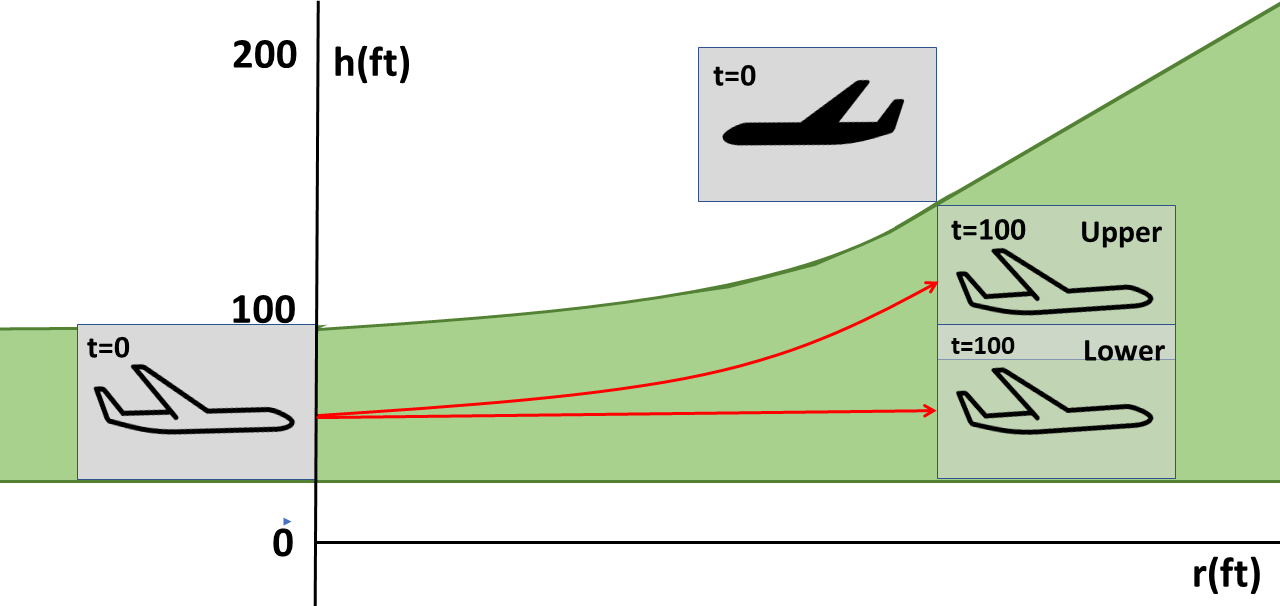}
    \Description[A two-sided safe region.]{An encounter between ownship and non-maneuvering intruder showing the two-sided region, with two compliant trajectories.}
    \caption{An encounter between ownship and non-maneuvering intruder showing the two-sided region, with two compliant trajectories in red.}
    \label{fig:twosidednointruder}
\end{figure}

Given the complexity of proving safeability, we first prove bounded-time, upper and lower bounded safe regions as building blocks to proving safeability. All previous safe regions constructed in this article have only been lower-bounded for $w=1$ (or only upper-bounded for $w=-1$), meaning the ownship must achieve some minimum acceleration or velocity in order to stay within the region. We now incorporate two-sided bounds on the ownship's acceleration and vertical velocity to construct a region which is bounded from above \emph{and} below (see \rref{fig:twosidednointruder}). This is important because in a safeable scenario, the ownship could receive a reversal advisory, meaning it must reverse its direction of flight between the first and second advisories. In this case, it is imperative that we know both an upper and lower bound on the trajectory of the ownship in the first phase of flight in order to reason about its position and velocity in the second phase of flight.

The following models prove safety of the two-sided safe regions only up to some time $\epsilon$. While it is \emph{not sufficient} only to prove safety for some bounded $\epsilon$ amount of time, this step lays important groundwork for proving the fully safeable region in \rref{sec:SafeNon}. The safeable proofs require that an advisory be provably safe only for $\epsilon$ time and that some followup advisory is safe indefinitely, so the bounded-time proof is a necessary stepping stone to ultimately proving the safeable region.   

We, again, start with the model which allows no maneuverability to the intruder before moving on to the model which allows the intruder to accelerate vertically.

\subsection{Model} \label{sec:BoundNonModel}
\rref{model:boundednon} can be found below with differences to the infinite-time \rref{model:infinitenon} highlighted in \textbf{bold}.
\setcounter{modelline}{0}
\begin{model}[b!]
\caption{Bounded-time safety for a non-maneuvering intruder}
\label{model:boundednon}
\begin{align*}
\text{init}     &\,\bigl|\mline{line:boundednon-init} \rp \geq0\land\hp>0\land r_v\geq0\land a_{\text{lo}}>0\land(w=-1\lor w=1) \bm{\land a_{\textbf{up}} > a_{\textbf{lo}}}\land {} \\
\text{R}        &\,\bigl|\mline{line:boundednon-R} \bm{C_{\textbf{impl}}^{\epsilon}(r,h,v,w,v_{\textbf{lo}},v_{\textbf{up}})} \\
                &\,\phantom{\bigl|}\mline{line:boundednon-imply} \limply  \\
                &\,\phantom{\bigl|}\mline{line:boundednon-loopstart}\bigl[\bigl( \\
\text{advisory} &\left|
\begin{aligned}
  &\mline{line:boundednon-adv}\quad\bigl((\pchoice{\pumod{w}{1}}{\pumod{w}{-1}}); \pumod{v_{\text{lo}}}{\ast};\bm{v_\textbf{up} := *;} \ptest{\bm{C_{\textbf{impl}}^{\epsilon}(r,h,v,w,v_{\textbf{lo}},v_{\textbf{up}})}};\\ &\mline{line:boundednon-adv-2}\phantom{\quad\bigl((\pchoice{\pumod{w}{1}}{\pumod{w}{-1}}); \pumod{v_{\text{lo}}}{\ast};\bm{v_\textbf{up} := *;}~} \pumod{\text{advisory}}{(w,v_{\text{lo}},\bm{v_{\textbf{up}}})}\bigr); 
 \end{aligned}
 \right.\\
                &\,\phantom{\bigl|}\mline{line:boundednon-timereset}\phantom{\bigl(}\quad\bm{t := 0;}  \\
\text{ownship}  &\,\bigl|\mline{line:boundednon-ownship}\phantom{\bigl(}\quad\bm{\bigl(}~ \pdual{\bigl(\prandom{a_o}; \ptest{(-a_\text{max} \leq a_o \leq a_\text{max})}\bigr)}; \\
\text{motion}   &\left|
\begin{aligned}
    &\mline{line:boundednon-motion-ode}\phantom{\bigl(}\quad\phantom{\bm{\bigl(}}~\{\D{r}=-r_v \syssep \D{h}=-v\syssep\D{v}=a_o \syssep \bm{t' = 1} \bm{~\&~ (t \leq \epsilon \lor \epsilon < 0)}\\
    &\mline{line:boundednon-motion-event}\phantom{\bigl(\bigl(}  \quad\qquad\qquad\qquad \bm{\land ((w a_o \leq 0 \lor wv \leq wv_\textbf{up}) \cup (w a_o \geq 0 \land wv \geq wv_\textbf{up}))} \}\\
\end{aligned}
\right.\\
&\,\phantom{\bigl|}\mline{line:boundednon-eventloopend}\phantom{\bigl(}\quad\bm{\bigr)^\ast}  \\
\neg\text{NMAC} &\,\Bigl|\mline{line:boundednon-safe} \phantom{\Bigl[\,}\bigr)^\ast\bigr]\bigl(\lvert r \rvert > \rp \lor \lvert h \rvert > \hp\bigr)
\end{align*}
\end{model}
First, we add symbolic upper bounds on the ownship acceleration $a_{\text{up}}$ and vertical velocity $v_{\text{up}}$ (usually $g/2$ and $10,000$\,ft/min, respectively). 
The upper bound $v_\text{up}$ increases the complexity of the ownship strategy: the strategy in \rref{thm:nonmaneuvering} picked acceleration $a_o \geq wa_\text{lo}$ towards satisfying $wv \geq wv_\text{lo}$, but now we must ultimately adjust this choice again before violating $wv \leq wv_\text{up}$.
We address this with an \emph{event-triggered} design in lines~\ref{line:boundednon-motion-ode}--\ref{line:boundednon-motion-event} with evolution domain constraints to monitor when $wv$ crosses $wv_\text{up}$.
In order to not discard behavior for the sake of detecting this event, we follow the standard pattern \cite{DBLP:journals/sttt/QueselMLAP16} to model event monitoring with overlapping evolution domain constraints.\footnote{The notation $\{\D{x}=f(x)~\&~P(x) \land (Q_1(x) \cup Q_2(x))\}$ is shorthand notation for the non-deterministic choice between ODEs that only differ in their evolution domain constraints: $\{\D{x}=f(x)~\&~P(x) \land Q_1(x)\} \cup \{\D{x}=f(x)~\&~P(x) \land Q_2(x)\}$.
Crucially, $\{\D{x}=f(x)~\&~P(x) \land (Q_1(x) \cup Q_2(x))\}$ differs from $\{\D{x}=f(x)~\&~P(x) \land (Q_1(x) \lor Q_2(x))\}$ in its ability to detect the events of handover from $Q_1(x)$ to $Q_2(x)$ or vice versa.
}
We get notified that an event occurred exactly at the boundary where the two evolution domain constraints $wa_o \leq 0 \lor wv \leq wv_\text{up}$ and $wa_o \geq 0 \land wv \geq wv_\text{up}$ overlap.
Through the loop on lines~\ref{line:boundednon-ownship}--\ref{line:boundednon-eventloopend}, the ownship can react to the event by selecting a new acceleration $a_o$ in \rref{line:boundednon-ownship}.

To model the bounded-time safety, we add variable $t$ which acts as a timer for the duration of the dynamics. It is reset to $0$ on \rref{line:boundednon-timereset} at the beginning of each new advisory and evolves at a constant rate in the dynamics equation on \rref{line:boundednon-motion-ode}. 
Once this variable reaches $\epsilon$, the differential equations are stopped from evolving further, and a new advisory must be issued.
This time bound $t\leq \epsilon$ guarantees that the ownship will be given a new advisory and update its acceleration after at most $\epsilon$-time has passed. 
The time-unbounded case is represented with the condition $\epsilon < 0$.

Another important difference introduced in this model is the removal of the test $\ptest{\ltrue}$ from \rref{line:boundednon-adv}. This has been removed to disallow the ownship from blindly continuing with the old advisory after $\epsilon$ time. While the same advisory can be given on each iteration of the loop, the removal of the test $\ptest{\ltrue}$ ensures that new advisories will definitely be considered after each $\epsilon$ time step. 

Even though a new advisory is given after $\epsilon$ time elapses, this does not mean that a safe advisory \emph{exists}. If no safe advisory exists, the test on \rref{line:boundednon-adv} would fail, and the model would be vacuously true since no possible runs of the system exist (and therefore all runs satisfy the postcondition). Thus, our model is only meaningful for $\epsilon$ time, after which we cannot draw any meaningful conclusions about safety. However, we address this issue of liveness in the models in Sections \ref{sec:SafeNon} and \ref{sec:SafeVert}.  

\subsection{Implicit Formulation of the Safe Region} \label{sec:BoundNonRegions}
The safe region of this model now consists of two separate safe regions, the lower bounded region $L^{\epsilon}_{\text{impl}}$, and the new upper region $U^{\epsilon}_{\text{impl}}$ listed in \rref{def:twosidedepsilon}.

The following observation is important for understanding the safe regions.
In the event that the ownship is already exceeding $v_{\text{up}}$ when receiving the $v_{\text{up}}$ advisory, it is unrealistic for the pilot to accelerate downwards to reach the upper bound $v_{\text{up}}$ from above. 
Instead, we assume the pilot will not accelerate further if she has already exceeded the target $v_{\text{up}}$.
Therefore, on an initial overcompliance in vertical velocity $(wv \geq wv_{\text{up}})$, the target velocity becomes $v$ as opposed to $v_{\text{up}}$. 

\noindent\paragraph{First case:}
Consider $w = + 1$ and $v_{\text{up}} \geq v$. Just like for the lower safe region, we consider a nominal trajectory $\mathcal{N}_{\mathit{up}}(t)$ to characterize the upper safe region. In this region, we again accelerate upwards at the upper acceleration $a_{\text{up}}$ until reaching $v_{\text{up}}$ (or $v$ in the case of initial overcompliance) and then continue linearly at velocity max$(v_{\text{up}},v)$. As before, we continue at a constant horizontal velocity $r_v$. We give the position along the nominal trajectory as a function of time:
\[
  (r_n,h_n) =
  \begin{cases}
        \left( r_vt, \frac{a_{\text{up}}}2t^2+vt\right) & \text{if $0 \leq t < \frac{v_{\text{up}}-v}{a_{\text{up}}}$} \\
        \left(r_vt,v_{\text{up}}t - \frac{(v_{\text{up}}-v)^2}{2a_{\text{up}}}\right) & \text{if $\frac{v_{\text{up}}-v}{a_{\text{up}}} \leq t$} \\
  \end{cases}
\]

In the same way that we knew that the ownship is above the nominal trajectory in the lower-bounded region, we know that the ownship will be below this upper nominal trajectory because $a_o \leq a_{\text{up}}$ and $v \leq v_{\text{up}}$ (or $a_o \leq 0$). Therefore, in order to ensure safety of the ownship we need:
\[
\forall t\, \forall r_n \, \forall h_n \left((t\leq \epsilon \lor \epsilon < 0) \land (r_n,h_n) \in \mathcal{N}_{\mathit{up}}(t) \rightarrow \vert r - r_n \vert > r_p \lor h - h_n > h_p\right)
\]
Now that we have the regions $L^{\epsilon}_{\text{impl}}$ and $U^{\epsilon}_{\text{impl}}$ (\rref{def:twosidedepsilon}), we can characterize the two-sided safe region $C^{\epsilon}_{\text{impl}}$ as their disjunction. Given the fact that $a_{\text{lo}} \leq a_o \leq a_{\text{up}}$, we know that the ownship stays between the two nominal trajectories, so its flight is properly upper and lower bounded; thus, as long as either the lower nominal trajectory or the upper nominal trajectory avoids collision, the ownship will be safe. This is the motivation for the disjunction of the two safe regions: since the ownship is between the two trajectories, only one trajectory needs to be safe for the ownship to be safe as well. The implicit formulation of $C^{\epsilon}_{\text{impl}}$ is shown in \rref{def:twosidedepsilon} and used in \rref{thm:bounded-non}, which has been verified to be safe using \KeYmaeraX.

\begin{definition}[Implicit formulation of the two-sided safe region]%
\label{def:twosidedepsilon}%
\begin{align*}
    T_{\text{up}}(v,w,v_\text{up}) &\equiv \frac{\text{max}(0,w(v_{\text{up}}-v))}{a_{\text{up}}}\\
    A_{\text{up}}(v,w,v_\text{up},h_n,t) &\equiv \left( 0 \leq t < T_{\text{up}} \land h_n = \frac{wa_{\text{up}}}2t^2 + vt\right)  \\
    &\lor \left( t \geq T_{\text{up}} \land h_n = w\text{max}(wv_{\text{up}},wv)t - \frac{w\text{max}(0,w(v_{\text{up}}-v))^2}{2a_{\text{up}}} \right) \\
    L_{\text{impl}}^{\epsilon}(r,h,v,w,v_\text{lo}) &\equiv \forall t \, \forall r_n \, \forall h_n \bigl( (t\leq \epsilon \lor \epsilon < 0) \land r_n = r_vt \land A_{\text{lo}}(v,w,v_\text{lo},h_n,t)\\ 
    & \phantom{\equiv \forall t \, \forall r_n \, \forall h_n ( (t\leq \epsilon \lor \epsilon < 0) \land} \limply (\vert r-r_n\vert >r_p \lor w(h_n-h) > h_p )\bigr)\\
    U_{\text{impl}}^{\epsilon}(r,h,v,w,v_\text{up}) &\equiv \forall t \, \forall r_n \, \forall h_n \bigl( (t\leq \epsilon \lor \epsilon < 0) \land r_n = r_vt \land A_{\text{up}}(v,w,v_\text{up},h_n,t)\\ 
    & \phantom{\equiv \forall t \, \forall r_n \, \forall h_n \bigl( (t\leq \epsilon \lor \epsilon < 0) \land} \limply (\vert r-r_n\vert >r_p \lor w(h-h_n) > h_p )\bigr)\\
    C^{\epsilon}_{\text{impl}}(r,h,v,w,v_\text{lo},v_\text{up}) &\equiv wv_\text{lo} \leq wv_\text{up} \land (L_{\text{impl}}^{\epsilon}(r,h,v,w,v_\text{lo})  \lor U_{\text{impl}}^{\epsilon}(r,h,v,w,v_\text{up}))
\end{align*}%
with $A_{\text{lo}}(v,w,v_\text{lo},h_n,t)$ and $T_{\text{lo}}(v,w,v_\text{lo})$ per \rref{def:implicitinfinite}.
\end{definition}

\begin{theorem}[Bounded-time non-maneuvering intruder: correctness of two-sided bounded-time safe regions]
\label{thm:bounded-non}
The \dGL formula given in \rref{model:boundednon} is valid. That is as long as the advisories obey formula $C^{\epsilon}_{\text{impl}}$ in \rref{def:twosidedepsilon} the winning strategy will avoid NMAC.
\end{theorem}

\begin{proof}
The \KeYmaeraX proof develops a winning strategy for choosing ownship control $a_o$:
\begin{equation*}
a_o = 
\begin{cases}
wa_\text{lo} & \text{if}~wv < wv_\text{lo} \\
0 & \text{if}~wv_\text{lo} \leq wv \leq wv_\text{up} \\
0 & \text{if}~wv_\text{up} < wv
\end{cases}
\end{equation*}
We pick $wa_\text{lo}$ to accelerate towards the advisory if the ownship is not yet in compliance $wv < wv_\text{lo}$.
This case crucially relies on the event-triggered design, which notifies the ownship of a required change in strategy before violating $wv \leq wv_\text{up}$.
When in compliance $wv_\text{lo} \leq wv \leq wv_\text{up}$ or in overcompliance $wv_\text{up} < wv$, we simply pick $a_\text{lo}=0$ to maintain the current climb rate.
The safe region $C^{\epsilon}_{\text{impl}}$ serves as a loop invariant.
\end{proof}

\section{Bounded-Time Safety for a Vertically-Maneuvering Intruder} \label{sec:BoundVert}

With the groundwork laid in \rref{sec:BoundNon}, we expand \rref{model:boundednon} for bounded-time safety to \rref{model:boundedvert} where we allow the intruder to control its vertical velocity.

\subsection{Model} \label{sec:BoundVertModel}
\setcounter{modelline}{0}
\begin{model}
\caption{Bounded-time vertically-maneuvering intruder}
\label{model:boundedvert}
\begin{align*} 
\text{init}     &\,\bigl|\mline{line:boundedvert-init} \rp \geq0\land\hp>0\land r_v\geq0\land a_{\text{lo}}>0\land a_{\text{up}} > a_{\text{lo}}\land\bm{c > 0} \,\land(w=-1\lor w=1)\land {} \\
\text{R}        &\,\bigl|\mline{line:boundedvert-R} C_{\text{impl}}^{\epsilon}(r,h,v,w,v_{\text{lo}},v_{\text{up}}) \\
                &\,\phantom{\bigl|}\mline{line:boundedvert-imply} \limply  \\
                &\,\phantom{\bigl|}\mline{line:boundedvert-loopstart}\bigl[\bigl( \\
\text{advisory} &\left|
\begin{aligned}
  &\mline{line:boundedvert-adv}\phantom{\bigl(}\quad\bigl((\pchoice{\pumod{w}{1}}{\pumod{w}{-1}}); \pumod{v_{\text{lo}}}{\ast};\pumod{v_{\text{up}}}{\ast}; \ptest{C_{\text{impl}}^{\epsilon}}; \pumod{\text{advisory}}{(w,v_{\text{lo}},v_{\text{up}}))}\bigr); 
 \end{aligned}
 \right.\\
                &\,\phantom{\bigl|}\mline{line:boundedvert-timereset}\phantom{\bigl(}\quad\bm{t := 0;}  \\
\text{estimate}  &\,\bigl|\mline{line:boundedvert-estimate}\phantom{\bigl(}\quad{\bigl(~\bm{\pdual{(c_o := \ast;\ptest{c_o \geq 0})}};} \\
\text{ownship}  &\,\bigl|\mline{line:boundedvert-ownship}\phantom{\bigl(\quad\bm{\bigl(~}}\bigl(~ \pdual{(\pumod{a_o}{\ast}; \ptest{(-a_\text{max} \leq a_o \leq a_\text{max})})}; \\
\text{intruder}  &\,\bigl|\mline{line:boundedvert-intruder}\phantom{\bigl(\quad\bm{\bigl(~}\bigl(~}\bigl(
\prandom{a_i};\ptest{(-\bm{c_o} < a_i < \bm{c_o})};\\
\text{motion}   &\left|
\begin{aligned}
    &\mline{line:boundedvert-motion-ode}\phantom{\bigl(\quad\bm{\bigl(~}\bigl(~\bigl(}\{\D{r}=-r_v \syssep \D{h}=-v\syssep\bm{v'=a_o-a_i} \syssep \D{t}=1 ~\&~ (t \leq \epsilon \lor \epsilon<0)\\
    &\mline{line:boundedvert-motion-events}\phantom{\bigl(\bigl(}  \quad\qquad\qquad\qquad  \land (\bm{wv \leq wv_\textbf{lo}} \cup \bm{wv_\textbf{lo} \leq wv \leq wv_\textbf{up}} \cup \bm{wv_\textbf{up} \leq wv}) \}  \\
\end{aligned}
\right.\\
&\,\phantom{\bigl|}\mline{line:boundedvert-odeloopend}\phantom{\bigl(\quad\bigl(\quad}\bigr)^\ast\\
&\,\phantom{\bigl|}\mline{line:boundedvert-aoloopend}\phantom{\bigl(\quad\bigl(}\bm{\bigr)^\ast}  \\
&\,\phantom{\bigl|}\mline{line:boundedvert-coloopend}\quad\bm{\bigr)^\ast}  \\
\neg\text{NMAC} &\,\Bigl|\mline{line:boundedvert-safe} \phantom{\Bigl[\,}\bigr)^\ast\bigr]\bigl(\lvert r \rvert > \rp \lor \lvert h \rvert > \hp\bigr)
\end{align*}
\end{model}

Again, the variables $t$ and $v_{\text{up}}$ as well as the constants $a_{\text{up}}$ and $\epsilon$ represent the upper safe region and our time bound, the variable $a_i$ tracks the intruder's acceleration, and we allow it to affect the relative vertical velocity $v$. 
The main update in \rref{model:boundedvert} is how the ownship reacts to intruder behavior.
In \rref{line:boundedvert-estimate}, the ownship estimates a bound $c_o$ for the upcoming intruder acceleration (e.g., the worst-case bound $c$ or a less permissive bound).
The ownship strategy can take this estimate into account when picking acceleration $a_o$.
The intruder then, in \rref{line:boundedvert-intruder}, gets to select intruder acceleration $a_i$: if that choice happens to fit to the ownship estimate, the test $\ptest{(-c_o < a_i < c_o)}$ passes and the model continues with motion on lines~\ref{line:boundedvert-motion-ode}--\ref{line:boundedvert-motion-events}. 
Otherwise the test fails and the only runs either have the intruder change its acceleration choice through the loop in lines~\ref{line:boundedvert-intruder}--\ref{line:boundedvert-motion-events}, or have control return to the ownship via the loop on lines~\ref{line:boundedvert-estimate}--\ref{line:boundedvert-coloopend} to update the estimate $c_o$.
This model allows for a variety of system implementations over a range of interaction requirements:
\begin{itemize}
    \item no interaction between intruder and ownship is required when the ownship uses the worst-case acceleration $c$ as its
    estimate $c_o$;
    \item the ownship may detect when the intruder acceleration exceeds the estimate $c_o$ and change its strategy in return;
    \item the ownship and intruder may cooperate to pick the bound $c_o$;
    \item the ownship may announce the bound $c_o$ as a requirement to the intruder.
\end{itemize}

The evolution domain constraint is again modified to detect when the relative climb rate falls below $wv_\text{lo}$ or exceeds the target $wv_\text{up}$, so that the ownship can change its strategy for overcoming the intruder motion.
The event detection mechanism of \rref{model:boundednon} is extended in \rref{line:boundedvert-motion-events} to detect when the ownship is about to no longer satisfy $wv_\text{lo} \leq wv \leq wv_\text{up}$, and so includes the choice between overlapping evolution domain constraints $wv \leq wv_\text{lo} \cup wv_\text{lo} \leq wv \leq wv_\text{up} \cup wv_\text{up} \leq wv$. 

Figure~\ref{fig:twosidedintruder} shows an example comparison of the current and former safe regions given the new intruder maneuverability. 
For simplicity we explain only the strategy for the upsense case $w=+1$ and the worst-case intruder acceleration $c$. 
In order to continue evolving, just as in the infinite-time \rref{model:infinitevert} from \rref{sec:InfVert} the ownship acceleration must be at least $a_{\text{lo}} + c$, or the ownship velocity is at least $v_{\text{lo}}$. 
Simultaneously, the ownship either initially overcomplies in terms of velocity ($wv_\text{up} \leq wv$), or its velocity must not exceed $v_{\text{up}}$ and its acceleration must not exceed $a_{\text{up}}-c$.
In this first case of overcompliance, it is no longer enough to stop accelerating upwards. 
The ownship now needs to compensate for the fact that the intruder could be accelerating downwards with acceleration at most $c$ by accelerating downwards with acceleration $c$ as well.  

\begin{figure}[tbhp]
    \centering
    \includegraphics[scale=\imagescale]{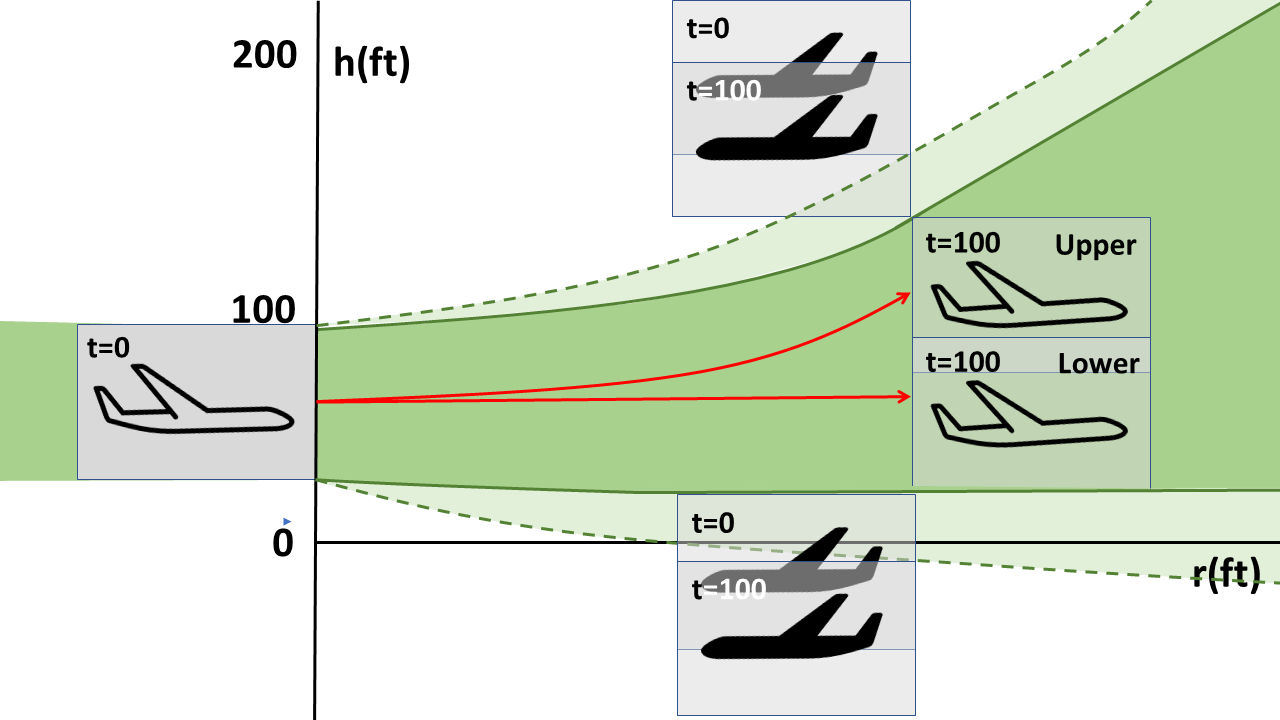}
    \Description[Nominal trajectory and two-sided safe region.]{An encounter between the ownship and a vertically-maneuvering intruder showing the two-sided region, with two compliant trajectories inside the safe region. The ownship is safe if it climbs strong enough to pass a climbing intruder above, or if it climbs weakly enough to stay below a descending intruder.}
    \caption{An encounter between ownship and vertically-maneuvering intruder showing the two-sided region, with two compliant trajectories in red. The larger region that would be safe without intruder maneuverability is shown in light green. }
    \label{fig:twosidedintruder}
\end{figure}

\subsection{Implicit Formulation of the Safe Region} \label{sec:BoundVertRegions}
The safe region of this model is exactly the same as that from \rref{sec:BoundNon}. Due to the full relativization of both the rate of vertical closure and vertical separation, as well as the requirement that $a_{\text{lo}} + c \leq a_o \leq a_{\text{up}} - c$, the $C^{\epsilon}_{\text{impl}}$ with coordinate system fixed at the intruder still applies to this model.

\begin{theorem}[Bounded-time vertically-maneuvering intruder: correctness of two-sided bounded-time safe regions]
\label{thm:bounded-vert}
The $\dGL$ formula given in \rref{model:boundedvert} is valid. That is as long as the advisories obey formula $C^{\epsilon}_{\text{impl}}$ from \rref{def:twosidedepsilon} the winning strategy will avoid NMAC.
\end{theorem}

\begin{proof}
The \KeYmaeraX proof develops a winning strategy for choosing ownship control $a_o$ in reaction to the worst-case intruder acceleration estimate $c_o=c$:
\begin{equation*}
    a_o =
    \begin{cases}
        w(a_\text{lo}+c) & \text{if}~wv \leq wv_\text{lo} \\
        0 & \text{if}~wv_\text{lo} \leq wv \leq wv_\text{up} \\
        -wc & \text{if}~wv_\text{up} \leq wv \\
    \end{cases}
\end{equation*}
If the ownship is not yet in compliance $wv \leq wv_\text{lo}$, We pick $a_o = w(a_\text{lo}+c)$, which is the minimum compliant acceleration $a_\text{lo}$ compensated for worst-case intruder acceleration $c$.
When in compliance $wv_\text{lo} \leq wv \leq wv_\text{up}$, we pick $a_o=0$ for proof simplicity, but any acceleration $a_o \leq w(a_\text{up}-c)$ that does not exceed the upper acceleration $a_\text{up}$ compensated for worst-case intruder acceleration $c$ would work as well.
Finally, in overcompliance $wv_\text{up} \leq wv$, we decelerate downwards with $a_o=-wc$ to compensate for worst-case intruder acceleration $-c$.
\end{proof}

Unlike in the non-maneuvering case \rref{thm:bounded-non}, in the presence of a vertically maneuvering intruder there exists no choice of $a_o$ that keeps the relative climb rate constant, since the intruder is allowed to change its acceleration $a_i$ arbitrarily often.
As a result, in \rref{thm:bounded-vert} the choice of $a_o=0$, and any other choice in $a_o \leq w(a_\text{up}-c)$, crucially relies on the event detection that informs the ownship when $wv_\text{lo} \leq wv \leq wv_\text{up}$ is about to be violated and a change in strategy is required.

\section{Safeability for a Non-Maneuvering Intruder} \label{sec:SafeNon}

In this section, we combine bounded-time safety with infinite-time safety to the notion of safeability: it is safe to follow a bounded-time region if there exists an infinite-time follow-up advisory.
The intuition behind the safeable region is shown in \rref{fig:safeable}. 
\begin{figure}[tbhp]
    \centering
    \includegraphics[scale=\imagescale]{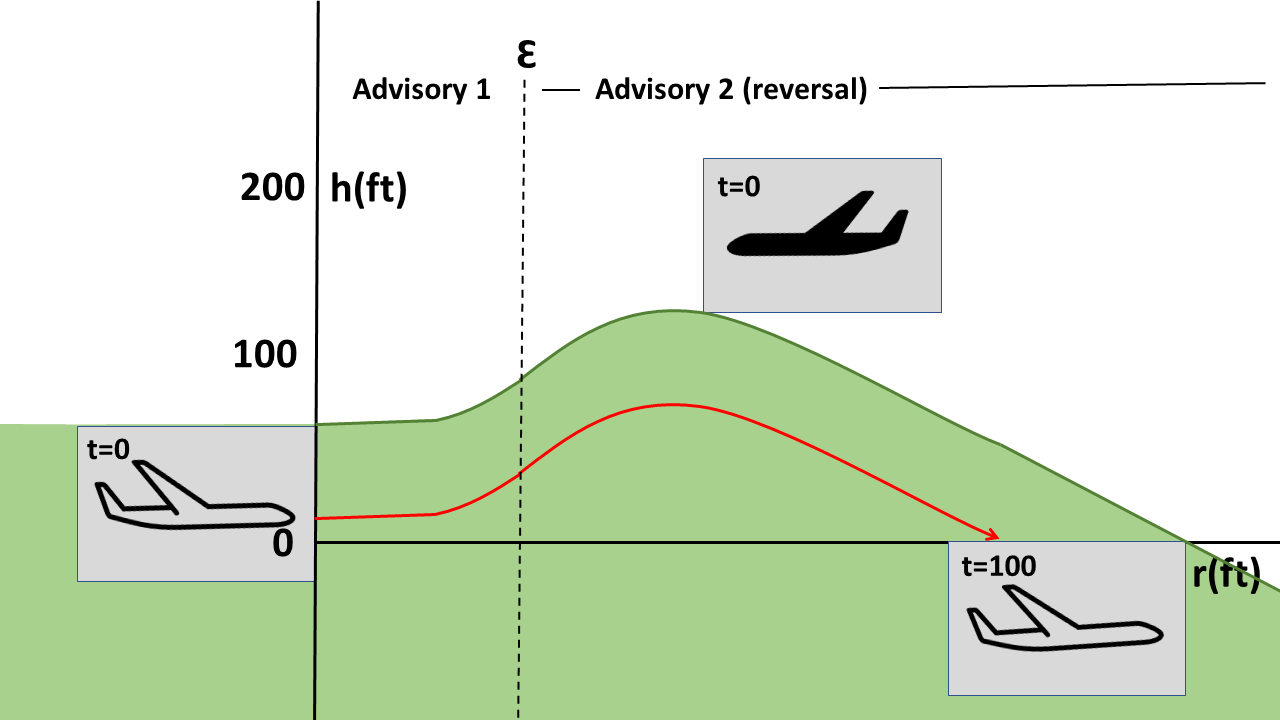}
    \Description[A safeable region with reversal.]{An encounter between the ownship and a non-moving intruder with reversal. The ownship is safe, if a follow-up advisory exists to avoid conflicts related to a potentially unsafe current advisory.}
    \caption{An encounter between ownship and intruder showing a safeable region with reversal.}
    \label{fig:safeable}
\end{figure}
We consider all the possible positions and speeds that the ownship could end up after $\epsilon$ time, in particular the lowest and highest such speeds and positions. At the lowest position, the most extreme strengthening is most critical, and at the highest position the most extreme reversal is most critical. The two safe regions achieved by these strengthenings and reversals give us the regions of ownship position which we can achieve by acting at time $\epsilon$. This is precisely the safeable region: the safe region achieved by following a weaker advisory until time $\epsilon$ and then following a stronger advisory indefinitely.  

\subsection{Model} \label{sec:SafeNonModel}
\rref{model:safeablenon} extends the bounded-time \rref{model:boundednon} with changes highlighted in \textbf{bold}.
\setcounter{modelline}{0}
\begin{model}
\caption{Safeability for a non-maneuvering intruder}
\label{model:safeablenon}
\begin{align*}
\text{init}     &\,\bigl|\mline{line:safeablenon-init} \rp \geq0\land\hp>0\land r_v\geq0\land a_{\text{lo}}>0\land(w=-1\lor w=1) \land a_{\mathit{up}} > a_{\mathit{lo}} + 2c \land \bm{\epsilon \geq 0} \land {} \\
\text{R}        &\,\bigl|\mline{line:safeablenon-r}\bm{C^{\textbf{safeable($\epsilon$)}}_{\textbf{impl}}(r,h,v,w,v_{\textbf{lo}},v_{\textbf{up}})} \\
                &\,\phantom{\bigl|}\mline{line:safeablenon-imply}\limply  \\
                &\,\phantom{\bigl|}\mline{line:safeablenon-loopstart}\bigl[\bigl( \\
\text{advisory} &\left|
\begin{aligned}
  &\mline{line:safeablenon-adv}\quad\bigl((\pchoice{\pumod{w}{1}}{\pumod{w}{-1}}); \pumod{v_{\text{lo}}}{\ast};v_{\text{up}} := *; \ptest{\bm{C^{\textbf{safeable($\epsilon$)}}_{\textbf{impl}}(r,h,v,w,v_{\textbf{lo}},v_{\textbf{up}})}};\\
  &\mline{line:safeablenon-adv-2}\phantom{\quad\bigl((\pchoice{\pumod{w}{1}}{\pumod{w}{-1}}); \pumod{v_{\text{lo}}}{\ast};v_{\text{up}} := *;~}\pumod{\text{advisory}}{(w,v_{\text{lo}},v_{\text{up}})}\bigr); 
 \end{aligned}
 \right.\\
                &\,\phantom{\bigl|}\mline{line:safeablenon-timereset}\phantom{\bigl(}\quad \pumod{t}{0};  \\
\text{ownship}  &\,\bigl|\mline{line:safeablenon-ownship}\phantom{\bigl(}\quad\bigl(~ \pdual{\bigl(\prandom{a_o}; \ptest{(-a_\text{max} \leq a_o \leq a_\text{max})}\bigr)}; \\
\text{motion}   &\left|
\begin{aligned}
    &\mline{line:safeablenon-motion-ode}\phantom{\bigl(}\quad\phantom{\bigl(}~\{\D{r}=-r_v \syssep \D{h}=-v\syssep\D{v}=a_o \syssep t' = 1 \bm{~\&~ \bm{t \leq \epsilon}}\\
    &\mline{line:safeablenon-motion-event}\phantom{\bigl(\bigl(}  \quad\qquad\qquad\qquad \land ((w a_o \leq 0 \lor wv \leq wv_\text{up}) \cup (w a_o \geq 0 \land wv \geq wv_\text{up})) \}\\
\end{aligned}
\right.\\
&\,\phantom{\bigl|}\mline{line:safeablenon-eventloopend}\phantom{\bigl(}\quad\bigr)^\ast  \\
\neg\text{NMAC} &\,\Bigl|\mline{line:safeablenon-safe} \phantom{\Bigl[\,}\bigr)^\ast\bigr](\lvert r \rvert > \rp \lor \lvert h \rvert > \hp)
\end{align*}
\end{model}
We now require $\epsilon$ to be positive so that the initial $\epsilon$-time region is finitely bounded, and we use region $C^{\text{safeable($\epsilon$)}}_{\text{impl}}$ on lines \ref{line:safeablenon-r} and \ref{line:safeablenon-adv}. 
The extensive changes to the safe region are discussed next.

\subsection{Implicit Formulation of the Safe Region} \label{sec:SafeNonRegions}
The region $C^{\text{safeable($\epsilon$)}}_{\text{impl}}$ again consists of a lower bound and an upper bound: regions $L^{\text{safeable($\epsilon$)}}_{\text{impl}}$ and $U^{\text{safeable($\epsilon$)}}_{\text{impl}}$ now combine two separate conditions: a region up to time $\epsilon$ and a region from time $\epsilon$ onward. 
Up to time $\epsilon$, we follow bounded-time $L^{\epsilon}_{\text{impl}}$ and $U^{\epsilon}_{\text{impl}}$ from \rref{def:twosidedepsilon}. 
In order to be allowed to issue an advisory, the region now encodes that from time $\epsilon$ onward there must exist a new advisory given the potential ownship velocity and vertical separation at time $\epsilon$ under which the infinite regions $L^{-1}_{\text{impl}}$ and $U^{-1}_{\text{impl}}$ are satisfied. 

Note that this model is similar to the previous bounded-time \rref{model:boundednon} but also proves liveness: after the $\epsilon$ time bound, we guarantee the existence of another advisory which keeps the ownship safe. The final region $C^{\text{safeable($\epsilon$)}}_{\text{impl}}$ is a disjunction of $L^{\text{safeable($\epsilon$)}}_{\text{impl}}$ and $U^{\text{safeable($\epsilon$)}}_{\text{impl}}$.

\begin{definition}[Implicit two-sided safeable region]%
\label{def:implicitsafeability}%
\allowdisplaybreaks%
\begin{equation*}
\begin{split}
    L^{\text{safeable}(\epsilon)}_{\text{impl}}&(r,h,v,w,v_\text{lo}) \equiv L^{\epsilon}_{\text{impl}}(r,h,v,w,v_{\text{lo}}) \land \\
    &\forall h^{\text{ex}}_L \forall v^{\text{ex}}_L \left( \left( 0 \leq \epsilon < T_{\text{lo}}(v,w,v_\text{lo}) \land   h^{\text{ex}}_L = \frac{wa_{\text{lo}}}2\epsilon^2 + v\epsilon \land v^{\text{ex}}_L = wa_{\text{lo}}\epsilon+v  \right. \right.\\
    &\phantom{\forall h^{\text{ex}}_L \forall v^{\text{ex}}_L \Bigl( \Bigl(}\left. \left. \lor~  \epsilon \geq T_{\text{lo}}(v,w,v_\text{lo})  \land h^{\text{ex}}_L = v^{\text{ex}}_L\epsilon - \frac{w\text{max}(0,w(v_{\text{lo}}-v))^2}{2a_{\text{lo}}}  \land  v^{\text{ex}}_L = v_{\text{lo}} \right) \right. \\
    &\phantom{\forall h^{\text{ex}}_L \forall v^{\text{ex}}_L \Bigl(} \left. \limply \exists v^{\text{ex}}_{\text{lo}} L^{-1}_{\text{impl}}(r-r_v\epsilon,h-h^{\text{ex}}_L,v^{\text{ex}}_L,w,v^{\text{ex}}_{\text{lo}})\right)
\end{split}
\end{equation*}
\begin{equation*}
    \begin{split}
    U^{\text{safeable}(\epsilon)}_{\text{impl}}&(r,h,v,w,v_\text{up}) \equiv U^{\epsilon}_{\text{impl}}(r,h,v,w,v_{\text{up}}) \land \\
   &\forall h^{\text{ex}}_U \forall v^{\text{ex}}_U \left( \Bigl( 0 \leq \epsilon < T_{\text{up}}(v,w,v_\text{up}) \land   h^{\text{ex}}_U = \frac{wa_{\text{up}}}2\epsilon^2 + v\epsilon \land v^{\text{ex}}_U = wa_{\text{up}}\epsilon+v \right.\\
    &\phantom{\forall h^{\text{ex}}_U \forall v^{\text{ex}}_U \Bigl( \Bigl( } \left. \lor~ \epsilon \geq T_{\text{up}}(v,w,v_\text{up}) \land h^{\text{ex}}_U = v^{\text{ex}}_U\epsilon - \frac{w\text{max}(0,w(v_{\text{up}}-v))^2}{2a_{\text{up}}} \right. \\ 
    &\phantom{\forall h^{\text{ex}}_U \forall v^{\text{ex}}_U \Bigl( \Bigl( ~\lor~ \epsilon \geq T_{\text{up}}(v,w,v_\text{up})} \left. \land   v^{\text{ex}}_U =  w\text{max}(wv_{\text{up}},wv)  \Bigr) \right.\\
    &\phantom{\forall h^{\text{ex}}_U \forall v^{\text{ex}}_U \Bigl(} \left. \limply \exists v^{\text{ex}}_{\text{up}}  L^{-1}_{\text{impl}}(r-r_v\epsilon,h-h^{\text{ex}}_U,v^{\text{ex}}_U,-w,v^{\text{ex}}_{\text{up}})\right)
    \\
    C^{\text{safeable($\epsilon$)}}_{\text{impl}}&(r,h,v,w,v_\text{lo},v_\text{up}) \equiv wv_\text{lo} \leq wv_\text{up} \land \left(L^{\text{safeable($\epsilon$)}}_{\text{impl}}(r,h,v,w,v_\text{lo})  \lor U^{\text{safeable($\epsilon$)}}_{\text{impl}}(r,h,v,w,v_\text{up}) \right)
    \end{split}
\end{equation*}
with $L^{-1}_{\text{impl}}$, $T_{\text{lo}}$ per \rref{def:implicitinfinite}, and $L^{\epsilon}_{\text{impl}}$, $U^{\epsilon}_{\text{impl}}$, $T_{\text{up}}$ per \rref{def:twosidedepsilon}.
\end{definition}

\begin{theorem}
\label{thm:safeablenon}
The \dGL formula given in \rref{model:safeablenon} is valid. That is as long as the advisories obey formula $C^{\text{safeable($\epsilon$)}}_{\text{impl}}(r,h,v,w,v_\text{lo},v_\text{up})$ in \rref{def:implicitsafeability} the winning strategy will avoid NMAC.
\end{theorem}

\begin{proof}
The \KeYmaeraX proof develops a winning strategy for choosing ownship control $a_o$:
\begin{equation*}
a_o = 
\begin{cases}
wa_\text{lo} & \text{if}~wv < wv_\text{lo} \\
0 & \text{if}~wv_\text{lo} \leq wv < wv_\text{up} \\
0 & \text{if}~wv_\text{up} \leq wv
\end{cases}
\end{equation*}
The proof is more involved, because the theorem is stronger, but the strategy is as in the bounded-time case.
We pick $wa_\text{lo}$ to accelerate towards the advisory if the ownship is not yet in compliance $wv < wv_\text{lo}$. We can keep this strategy until the event-trigger that the ownship has reached the advisory $wv_\text{lo}$. 
When $wv_\text{lo} \leq wv < wv_\text{up}$, we pick $a_o=0$ so that $wv$ stays within $wv_\text{lo}$ and $wv_\text{up}$. The ownship can choose an advisory that does not keep it within this range indefinitely, but it then relies on the event-trigger to notify when it is no longer in compliance,
necessitating a new choice of acceleration. In the overcompliance case, we again choose $0$ for $wa_\text{lo}$. The safe region $C^{\text{safeable($\epsilon$)}}_{\text{impl}}$ serves as a loop invariant. Due to this invariant, we know the existence of a future advisory such that the 
ownship will be safe indefinitely after using these strategies for the initial $\epsilon$ time. 
\end{proof}

\section{Safeability for a Vertically-Maneuvering Intruder} \label{sec:SafeVert}
Taking all of the groundwork laid by the previous sections, we finally present the model and safeable regions generalized for the case where the intruder can change its vertical acceleration. 

\subsection{Model}
\rref{model:safeablevert} extends bounded-time \rref{model:boundedvert} with changes highlighted in \textbf{bold}.
\setcounter{modelline}{0}
\begin{model}[tbhp]
\caption{Safeability for a vertically-maneuvering intruder}
\label{model:safeablevert}
\begin{align*}
\text{init}   &\left|
\begin{aligned}
    &\mline{line:safeablevert-init1} \rp \geq0\land\hp>0\land r_v\geq0 \land c > 0 \land a_{\text{lo}}>0\land(w=-1\lor w=1) \land {} \\
    &\mline{line:safeablevert-init2} a_{\textbf{up}} > a_{\textbf{lo}} + 2c \land a_\text{max} \geq a_\text{lo}+c  \land \epsilon \geq 0 \land {} \\
\end{aligned}
\right.\\
\text{R}        &\,\bigl|\mline{line:safeablevert-r} \bm{C_{\textbf{impl}}^{\textbf{safeable}(\epsilon)}(r,h,v,w,v_{\textbf{lo}},v_{\textbf{up}})} \\
                &\,\phantom{\bigl|}\mline{line:safeablevert-imply}\limply  \\
                &\,\phantom{\bigl|}\mline{line:safeablevert-loopstart}\bigl[\bigl( \\
\text{advisory} &\left|
\begin{aligned}
  &\mline{line:safeablevert-adv}\quad\bigl((\pchoice{\pumod{w}{1}}{\pumod{w}{-1}}); \pumod{v_{\text{lo}}}{\ast};v_{\text{up}} := *; \bm{\ptest{C_{\textbf{impl}}^{\textbf{safeable}(\epsilon)}(r,h,v,w,v_{\textbf{lo}},v_{\textbf{up}})}};\\
  &\mline{line:safeablevert-adv-2}\phantom{\quad\bigl((\pchoice{\pumod{w}{1}}{\pumod{w}{-1}}); \pumod{v_{\text{lo}}}{\ast};v_{\text{up}} := *;~} \pumod{\text{advisory}}{(w,v_{\text{lo}},v_{\text{up}})}\bigr); 
 \end{aligned}
 \right.\\
                &\,\phantom{\bigl|}\mline{line:safeablevert-timereset}\phantom{\bigl(}\quad \pumod{t}{0};  \\
\text{estimate}  &\,\bigl|\mline{line:safeablevert-estimate}\phantom{\bigl(}\quad\bigl(~\pdual{(c_o := \ast;\ptest{c_o \geq 0})}; \\
\text{ownship}  &\,\bigl|\mline{line:safeablevert-ownship}\phantom{\bigl(\quad\bigl(~}\bigl(~ \pdual{(\pumod{a_o}{\ast}; \ptest{(-a_\text{max} \leq a_o \leq a_\text{max})})}; \\
\text{intruder}  &\,\bigl|\mline{line:safeablevert-intruder}\phantom{\bigl(\quad\bigl(~\bigl(~}\bigl(
\prandom{a_i};\ptest{(-c_o < a_i < c_o)};\\
\text{motion}   &\left|
\begin{aligned}
    &\mline{line:safeablevert-motion-ode}\phantom{\bigl(\quad\bigl(~\bigl(~\bigl(}\{\D{r}=-r_v \syssep \D{h}=-v\syssep \D{v}=a_o-a_i \syssep \D{t}=1 ~\&~ \bm{t \leq \epsilon}\\
    &\mline{line:safeablevert-motion-events}\phantom{\bigl(\bigl(}  \quad\qquad\qquad\qquad  \land (wv \leq wv_\text{lo} \cup wv_\text{lo} \leq wv \leq wv_\text{up} \cup wv_\text{up} \leq wv) \}  \\
\end{aligned}
\right.\\
&\,\phantom{\bigl|}\mline{line:safeablevert-odeloopend}\phantom{\bigl(\quad\bigl(\quad}\bigr)^\ast\\
&\,\phantom{\bigl|}\mline{line:safeablevert-aoloopend}\phantom{\bigl(\quad\bigl(}\bigr)^\ast  \\
&\,\phantom{\bigl|}\mline{line:safeablevert-coloopend}\quad\bigr)^\ast  \\
\neg\text{NMAC} &\,\Bigl|\mline{line:safeablevert-safe} \phantom{\Bigl[\,}\bigr)^\ast\bigr]\bigl(\lvert r \rvert > \rp \lor \lvert h \rvert > \hp\bigr)
\end{align*}
\end{model}
Again, we represent the intruder choice for $a_i$ at \rref{line:safeablevert-intruder} and the bound $c_o$ on this choice, as well as the loop around the intruder choice and dynamics to reflect that the intruder can change its acceleration throughout the encounter without the ownship being alerted to these changes. 

\subtitle{Implicit formulation of the safe region}
The safe region for \rref{model:safeablevert} follows \rref{def:implicitsafeability} due to the relativization of the vertical rate of closure and vertical separation.

\begin{theorem}
\label{thm:safeablevert}
The \dGL formula given in \rref{model:safeablevert} is valid. That is as long as the advisories obey formula $C^{\text{safeable($\epsilon$)}}_{\text{impl}}(r,h,v,w,v_\text{lo},v_\text{up})$ from \rref{def:implicitsafeability} the winning strategy will avoid NMAC.
\end{theorem}

\begin{proof}
The \KeYmaeraX proof develops a winning strategy for choosing ownship control $a_o$ in reaction to the worst-case intruder acceleration estimate $c_o=c$:
\begin{equation*}
    a_o =
    \begin{cases}
        w(a_\text{lo}+c) & \text{if}~wv \leq wv_\text{lo} \\
        0 & \text{if}~wv_\text{lo} \leq wv \leq wv_\text{up} \\
        -wc & \text{if}~wv_\text{up} \leq wv \\
    \end{cases}
\end{equation*}
This strategy matches that of the bounded time case. Again, if the ownship is not yet in compliance $wv \leq wv_\text{lo}$, we pick $a_o = w(a_\text{lo}+c)$. When the
ownship is compliant, with $wv_\text{lo} \leq wv \leq wv_\text{up}$, we pick $a_o=0$, and in the overcompliance case, we accelerate downwards with $a_o=-wc$. 
With the safe region $C^{\text{safeable($\epsilon$)}}_{\text{impl}}$ as the loop invariant, we know that a future advisory exists which will keep the ownship safe indefinitely, 
assuming the ownship follows these strategies up to $\epsilon$ time. 
\end{proof}

\section{Discussion} \label{sec:discussion}

Each of the proofs of the theorems presented in this article were completed in \KeYmaeraX with support of Mathematica's implementation of real quantifier elimination~\cite{mathematica}.
Besides designing the hybrid game model and identifying the safe regions for its advisories, the main insights in the proofs are the invariants and winning strategies.
The verified models crucially extends previous work with a rich representation of intruder behavior and its potentially adversarial nature.
The benefit of attributing different actions to different players in the game is that we can now represent interactions between the ownship and the intruder, including worst-case non-cooperative interactions, as well as interactions in which the intruder and ownship react to one another.

Hybrid games verification, however, leads to even more complicated real arithmetic decision problems, which, even if decidable, may need time-consuming, handmade proof simplification for the proof to close. The arithmetic complexity of the safe regions required an intimate knowledge in order to simplify the proof into pieces digestible to Mathematica. This was particularly relevant compared to hybrid systems ACAS~X verification~\cite{acasx} due to the addition of quantifier alternation stemming from games and new variables and constants to represent the intruder's maneuverability, as well as the introduction of more relationships between these values in the preconditions.
The majority of manual arithmetic simplifications are case splitting (e.g., between $w=-1$ or $w=1$), providing witnesses to quantifiers according to the winning strategy to help solvers, abbreviating terms to make inequality transitivity obvious, and selecting the relevant assumptions from the list of all assumptions.
\rref{tab:proofstatistics} compares our proofs in terms of tactic size as an indicator for relative manual proof effort, and proof checking duration as an indicator of arithmetic complexity.
\begin{table}[t]
    \caption{Proof statistics: tactic size and total proof checking duration, duration of all proof attempts at real arithmetic proof obligations (column QE), and duration of the successful attempts (column RCF)}
    \label{tab:proofstatistics}
    \centering
    \begin{tabular}{llccrrrr}
        \toprule
        & \textbf{Model} & \multicolumn{2}{c}{\bf Intruder} & \textbf{Tactic steps} & \multicolumn{3}{c}{\textbf{Proof checking} [s]}\\
        \cmidrule{3-4} \cmidrule{6-8}
        & & Vert. & Horiz. & & Total & QE & RCF \\
        \midrule
        \multirow{3}{*}{\footnotesize Infinite} & \rref{model:infinitenon}, \rref{thm:nonmaneuvering} & & & 163 & 18 & 15 & 12 \\
        & \rref{model:infinitevert}, \rref{thm:infinitevert} & $\checkmark$ & & 196 & 35 & 30 & 26 \\
        & \rref{model:infinitehoriz}, \rref{thm:implicitinfinitehoriz} & & $\checkmark$ & 242 & 35 & 29 & 23 \\
        \midrule
        \multirow{2}{*}{\footnotesize Bounded} & \rref{model:boundednon}, \rref{thm:bounded-non} & & & 540 & 112 & 97 & 79 \\
        & \rref{model:boundedvert}, \rref{thm:bounded-vert} & $\checkmark$ & & 724 & 319 & 298 & 177 \\
        \midrule
        \multirow{2}{*}{\footnotesize Safeable} & \rref{model:safeablenon}, \rref{thm:safeablenon} & & & 3402 & 1786 & 886 & 279 \\
        & \rref{model:safeablevert}, \rref{thm:safeablevert} & $\checkmark$ & & 2782 & 1404 & 1163 & 247 \\
        \bottomrule
    \end{tabular}
\end{table}
Its content has to be taken with a grain of salt, because there is considerable freedom in designing proofs and in making the tradeoff between manual steps and duration of arithmetic proof obligations.
We observe a few general trends across models:
\begin{itemize}
    \item Even with manual simplification, the real arithmetic proof obligations are responsible for a considerable portion of the proof-checking duration;
    \item Infinite-time models are considerably easier (smaller tactic size, faster proof) than bounded-time models, which are in turn considerably easier than safeable models;
    \item Intruder maneuverability increases proof complexity, but experience from controlling the branching when proving the many ways of ownship and intruder interaction in the non-maneuvering safeable model helped finding a smaller tactic design and faster proof in the vertical safeable model;
    \item The large number of combinations of ownship and intruder interactions in the safeable models makes it infeasible to manually simplify all the arithmetic proof obligations, which results in considerably longer proof checking duration; the large difference in the proofs of \rref{model:safeablenon} and \rref{model:safeablevert} between the duration of all proof attempts in the tactic QE and the duration of its successful attempt in column RCF indicates a potential for improving tactic heuristics or parallelization.
\end{itemize}

The increased complexity that each new variable introduces motivated our decision to relativize the rates of closure. This simplified the model and arithmetic, and allowed for the same region from the corresponding non-maneuvering intruder to be used. We cannot further relativize the models beyond the rates of closure, however; it is important that each actor in the encounter have sole control over their own accelerations independently to properly model that both actors can affect the outcome of an encounter in this aircraft game scenario. 

In terms of proof construction, the main challenge was in the development of the winning strategies for the ownship. This is a particular challenge in proving hybrid games; the ownship must pick a strategy without any knowledge of what the intruder may end up choosing, only knowledge of the broad limitations of the intruder's maneuvering capabilities. In the context of ACAS~X, this is complicated by the fact that the ownship does not have access to the acceleration of the intruder, so it cannot react to the intruder's choice of acceleration to update its own choice. Therefore, it was imperative to have a proper strategy in choosing the pilot's acceleration in order to prove our intruder-maneuverability models. 

A final challenge in hybrid game modeling is in the correct assignment of player responsibility for a given action. It is vital to the fidelity of the model that the choices in the model be resolved by the correct player to prevent one player having an unfair advantage. For instance, if the choice of advisory were resolved by the ownship, the ownship could choose an optimal advisory to follow. 
However, the choice of advisory needs to be resolved by the intruder because that proves that \textit{all} of the advisories which satisfy the safe region can be chosen and are shown safe.

\section{Related Work} \label{sec:relatedwork}

Kochenderfer and Chryssanthacopoulos \cite{MDP} design the ACAS~X lookup tables, the verification of which motivates this work (see \rref{sec:ACASXSec}). Von Essen and Giannakopoulou \cite{vonessen} use probabilistic model-checking in their analysis of a similar Markov Decision Process  \cite{MDP}, to explore the probability of occurrence of various unfavorable events. The outcome is limited due to its discretization of the continuous dynamics in analyzing the system and the implausible assumption that the intruder follows a random walk in the decision space. 

Holland \emph{et al.} \cite{sim1} and Chludzinski \cite{sim2} simulate encounters, many of which come from recorded flight data. Lee, \emph{et al.} \cite{lee} develops a technique called differential adaptive stress testing to find scenarios in which TCAS does not result in an NMAC, but ACAS~X does. These simulations provide interesting evaluations of the performance of ACAS~X, but only explore a finite set of the state space and therefore cannot allow any conclusions about the infinitely many other possible behaviors. 

Julian and Kochenderfer \cite{DDN1} train a deep neural network to approximate the ACAS~X lookup table to reduce the storage needs and  runtime of the system, and Irfan, \textit{et al.} \cite{DDN2} and Julian and Kochenderfer \cite{DDN3} explore applying formal methods to verify such neural networks. One drawback to this approach is that SMT does not support continuous dynamics, and all queries to the SMT solver must be in the form of discrete, linear regions. The inherent nonlinearity of the relevant regions when verifying ACAS~X severely limits the verifying ability of this approach.

Kouskoulas, \emph{et al.} \cite{jbj} develop a formally verified, quantifier-free predicate which, given a sequence of timed ownship and intruder maneuvers, checks whether or not an NMAC may occur. They do so by establishing envelopes around either aircraft that contain all altitudes reachable by each aircraft through each maneuver. While this work verifies the safety of pre-determined encounters between an ownship and intruder accelerating non-deterministically, it requires knowledge of the full sequence of maneuvers which the intruder will perform beforehand. Our work instead proves safety assuming no knowledge of the sequence of intruder maneuvers, better capturing the uncertain adversarial abilities of the intruder. 

Lygeros and Lynch \cite{lygeros} explore verification of the conflict resolution algorithms used in TCAS, a predecessor of ACAS~X, with hybrid techniques. This work is limited in its overzealous use of assumptions, for instance in assuming that two aircraft both using TCAS will ultimately be given opposite advisories in an encounter. Our work does not make assumptions about the actions of the intruder and ownship relative to one another; the decisions of one aircraft are independent from those of the other as are best expressed with game models.

Tomlin \emph{et al.} \cite{tomlin} presents a methodology with which to develop safe collision-avoidance maneuvers using hybrid systems. Platzer and Clarke \cite{clarke}, Loos \emph{et al.} \cite{loos} and Ghorbal \emph{et al.} \cite{ghorbal} also use hybrid systems to design and verify their own horizontal collision avoidance maneuvers. Dowek \emph{et al.} \cite{dowek} and Galdino \emph{et al.} \cite{galdino} design their own algorithms for collision avoidance, known as KB3D and KB2D, respectively, and verify their geometry using the PVS theorem prover. 

Other approaches for hybrid games have limited real-world applications due to the overly-restrictive assumptions that they place on the systems which they represent. Henzinger \emph{et al.} \cite{rectangle} work with rectangular hybrid games, which require strict upper and lower bounds on the continuous dynamics of the system and forgetful transitions. This work is one of the first case studies for hybrid games verification, the only example that we are aware of coming from Quesel and Platzer \cite{robot} in their feasibility study involving an abstract robot in a factory.

\section{Conclusion and Future Work} \label{sec:conclusion}
We applied hybrid games to the verification of ACAS~X to prove that under limited horizontal and vertical maneuverability of an adversarial intruder, an ownship given a safe advisory from ACAS~X always has a strategy to find a trajectory which avoids an NMAC despite subsequent intruder maneuvering. This work employed the principle of our previous work \cite{acasx} to identify regions of safety, whereby following any advisory in the safe region has an ownship strategy that will keep the aircraft clear of an NMAC. The safe regions and intruder capabilities are symbolic such that these models can be reused for future versions of ACAS~X and apply to any intruder which is less maneuverable than the ownship.
While increased arithmetic and proof complexity is the downside of working with hybrid games, the advantage is the significantly increased resulting predictive power, because collision freedom can be proved even if the intruder is maneuvering, which it will in reality.
In future work, we plan to explore verifying more complex ownship maneuvers in the horizontal direction and apply horizontal intruder maneuverability to the safeable model.

\begin{acks}
This research was sponsored by the \grantsponsor{afosr}{AFOSR}{} under grant number \grantnum{afosr}{FA9550-16-1-0288}.
The views and conclusions contained in this document are those of the authors and should not be interpreted as representing the official policies, either expressed or implied, of any sponsoring institution, the U.S. government or any other entity.
\end{acks}

\bibliographystyle{ACM-Reference-Format}
\bibliography{acasxgames-arxiv}

\end{document}